\documentclass[11pt]{article}
\usepackage[utf8]{inputenc}
\usepackage[colorlinks=true]{hyperref}
\hypersetup{urlcolor=blue, citecolor=red, linkcolor=blue}

\usepackage{colortbl}

\usepackage{graphicx}
\usepackage[active]{srcltx}

\usepackage{amsmath,amssymb}
\usepackage{bm}
\usepackage{enumerate}
\usepackage{amsthm}
\usepackage[all,2cell]{xy}
\usepackage{mathrsfs}
\usepackage{mathtools}
\mathtoolsset{showonlyrefs}

\usepackage{tikz-cd}
\usetikzlibrary{cd}

\usepackage{xcolor}
\DeclareUnicodeCharacter{0301}{*************************************}

\hypersetup{urlcolor=blue, citecolor=red, linkcolor=blue}

\usepackage[a4paper, left=20mm, right=20mm, top=25mm, bottom=25mm]{geometry}

\usepackage{marginnote}

\usepackage{imakeidx}
\makeindex[intoc]

\usepackage[nottoc]{tocbibind}
\usepackage[colorinlistoftodos]{todonotes} 

\theoremstyle{plain}
\newtheorem{theorem}{Theorem}[section]
\newtheorem{proposition}[theorem]{Proposition}
\newtheorem{corollary}[theorem]{Corollary}
\newtheorem{lemma}[theorem]{Lemma}

\theoremstyle{definition}
\newtheorem{definition}[theorem]{Definition}
\newtheorem{remark}[theorem]{Remark}
\newtheorem{example}[theorem]{Example}

\newcommand\restr[2]{{
  \left.\kern-\nulldelimiterspace 
  #1 
  \right|_{#2} 
}}

\newcommand{\R}{\mathbb{R}}

\renewcommand{\d}{\mathrm{d}}

\newcommand{\df}{\Omega}
\newcommand{\Cinfty}{\mathscr{C}^\infty}
\newcommand{\T}{{\mathrm T}}
\let\Tan\T
\newcommand{\cT}{\T^{\ast}}
\newcommand{\Id}{\mathrm{Id}}

\newcommand{\Lie}{\mathscr{L}}
\newcommand{\vf}{\mathfrak{X}}
\newcommand{\X}{\mathfrak{X}}

\newcommand{\F}{\mathcal{F}}

\newcommand{\parder}[2]{\frac{\partial #1}{\partial #2}}

\newcommand{\tparder}[2]{\partial #1/\partial #2}
\newcommand{\parderr}[3]{\frac{\partial^2 #1}{\partial #2\partial #3}}

\DeclareMathAlphabet{\mathpzc}{OT1}{pzc}{m}{it}
\def\d{\mathrm{d}}

\newcommand{\fracp}[2]{\frac{\partial #1}{\partial #2}}

\def\X{\mathfrak{X}}
\def\d{\mathrm{d}}

\def\Ker{\mathrm{Ker \,}}
\def\Tan{\mathrm{T}}
\def\Lie{\mathscr{L}}

\def\FL{\mathcal{F} L}

\newcommand{\totder}[2]{\frac{\mathrm{d} #1}{\mathrm{d} #2}}

\usepackage{fancyhdr}

\begin{document}
\parskip=3pt

\vspace{5em}

{\huge\sffamily\raggedright
The evolution operator connecting the\\[1ex] 
Lagrangian and Hamiltonian formalisms for contact systems
}

\vspace{1.5em}

{\large\raggedright
    \today
}

\vspace{1.5em}

{\Large\raggedright\sffamily
    Xavier Gràcia
}\vspace{1mm}\newline
{\raggedright
    Department of Mathematics, Universitat Politècnica de Catalunya\\
    Barcelona, Catalonia, Spain\\
    e-mail: 
    \href{mailto:xavier.gracia@upc.edu}{xavier.gracia@upc.edu} --- {\sc orcid}: 
    \href{https://orcid.org/0000-0003-1006-4086}{0000-0003-1006-4086}
}

\bigskip

{\Large\raggedright\sffamily
    Ángel Martínez-Muñoz
}\vspace{1mm}\newline
{\raggedright
    Department of Computer Engineering and Mathematics, Universitat Rovira i Virgili\\
    Avinguda Països Catalans 26, 43007 Tarragona, Spain\\
    e-mail: \href{mailto:angel.martinezm@urv.cat}{angel.martinezm@urv.cat} --- {\sc orcid}: \href{https://orcid.org/0009-0002-8944-5403}{0009-0002-8944-5403}
}

\bigskip

{\Large\raggedright\sffamily
    Xavier Rivas
}\vspace{1mm}\newline
{\raggedright
    Department of Computer Engineering and Mathematics, Universitat Rovira i Virgili\\
    Avinguda Països Catalans 26, 43007 Tarragona, Spain\\
    e-mail: \href{mailto:xavier.rivas@urv.cat}{xavier.rivas@urv.cat} --- {\sc orcid}: \href{https://orcid.org/0000-0002-4175-5157}{0000-0002-4175-5157}
}

\bigskip

{\Large\raggedright\sffamily
    Narciso Román-Roy
}\vspace{1mm}\newline
{\raggedright
    Department of Mathematics, Universitat Politècnica de Catalunya\\
    Barcelona, Catalonia, Spain\\
    e-mail: 
    \href{mailto:narciso.roman@upc.edu}{narciso.roman@upc.edu} --- {\sc orcid}: 
    \href{https://orcid.org/0000-0003-3663-9861}{0000-0003-3663-9861}
}

\vspace{2em}

{\large\bf\raggedright
    Abstract
}\vspace{1mm}\newline
{
\parindent 0pt

Some mechanical systems with dissipation can be described within the framework of the so-called contact mechanics: a modified form of the Euler--Lagrange equations stemming from Herglotz's variational principle, which admits a geometric formulation in terms of contact geometry. On the other hand, the study of singular Lagrangian systems and Dirac's theory of constraints can be enhanced by using the evolution operator~$K$ that connects the Lagrangian and Hamiltonian formalisms.

The main purpose of this paper is to transpose this evolution operator to the case of contact mechanics, and to study some of its main properties. In particular, we show that it provides a geometric description of the evolution equations and it relates the Hamiltonian and Lagrangian constraints. To illustrate the theory, we provide examples of singular contact systems based on modified versions of the simple pendulum and the Cawley Lagrangian.
}
\bigskip

\medskip


\bigskip


\newpage

{\setcounter{tocdepth}{2}
\def\baselinestretch{1}
\small
\def\addvspace#1{\vskip 1pt}
\parskip 0pt plus 0.1mm
\tableofcontents
}

\pagestyle{fancy}

\fancyhead[L]{Evolution operator for contact systems}    
\fancyhead[C]{}           
\fancyhead[R]{X. Gràcia, Á. Martínez-Muñoz, X. Rivas, N. Román-Roy}       

\fancyfoot[L]{}     
\fancyfoot[C]{\thepage}                  
\fancyfoot[R]{}            

\setlength{\headheight}{17pt}

\renewcommand{\headrulewidth}{0.1pt}  
\renewcommand{\footrulewidth}{0pt}    

\renewcommand{\headrule}{%
    \vspace{3pt}                
    \hrule width\headwidth height 0.4pt 
    \vspace{0pt}                
}

\setlength{\headsep}{30pt}  

\section{Introduction}

In analytical mechanics it is well known that
a regular Lagrangian function leads to uniquely determined dynamics in the velocity space $\T Q$ (Lagrangian formalism)
and in the phase space $\T^*Q$ (Hamiltonian formalism),
which are equivalent through the Legendre transformation.
This is no longer true when the Lagrangian is singular,
which is the case for relativistic invariant theories.
Moreover, a Hamiltonian formulation of mechanics or field theory is highly convenient if one aims for a quantum theory. 
This led P.A.M. Dirac and P.G. Bergmann,
near 1950
\cite{Ber_51,Dir_50},
to start the study of the Hamiltonian formulation of singular Lagrangian theories,
introducing the fundamental notions of constraints and gauge freedom.

Differential equations can be written for singular Lagrangian and Hamiltonian dynamics, 
but there is no guarantee that their solutions are unique,
upon fixation of initial conditions.
What is remarkable, however, 
is that there is still a kind of well-defined evolution,
but at the price of involving two different spaces.
Let us explain this point.

If a time evolution on a manifold $M$ is ruled by a differential equation defined by a vector field~$X$,
the dynamical trajectories obey
$\gamma' = X \circ \gamma$,
in coordinates $\dot x^i = X^i(x(t))$,
then the evolution of any function
(an observable, in physical terminology) 
can be computed based on the chain rule,
and is given by $\mathrm{D}(f \circ \gamma) = (\Lie_X f) \circ \gamma$,
in coordinates
$
\totder{f}{t} = 
\parder{f}{x^i} \dot x^i = \parder{f}{x^i} X^i
$.

Now consider a Lagrangian system given by a Lagrangian function $L \colon \T Q \to \R$.
Using natural coordinates $(q^i,v^i)$ in $\T Q$ and
$(q^i,p_i)$ in $\T^*Q$,
the definition of the Legendre map
$\FL \colon \T Q \to \T^*Q$,
in coordinates
$p_i = \parder{L}{v^i}$,
and the Euler--Lagrange equations,
$
\totder{}{t} \left( \parder{L}{v^i} \right) = \parder{L}{q^i}
$,
give us 
$$
\dot q^i = v^i ,
\qquad 
\dot p_i = \parder{L}{q^i}
\,.
$$
So, using the chain rule, we can compute the time derivative of an arbitrary function $g(q,p)$,
and obtain an expression like
$$
\totder{g}{t} = 
\parder{g}{q^i} v^i +
\parder{g}{p_i} \parder{L}{q^i}
\,.
$$
That is: 
even if we do not know the time-evolution of a function
$g(q,p)$ on the phase space,
we \emph{know it} as expressed in terms of functions 
on the velocity space.
Maybe the reader has noticed that we loosely mixed 
velocity and momentum coordinates
(as, for instance, P.A.M. Dirac did in his 1950 paper,
or K. Kamimura in 1982
\cite{Kam_82}).
In geometric terms, one may work on the direct sum 
$\T^*Q \oplus \T Q$
(sometimes called the Pontryagin bundle),
see, for instance, 
\cite{Ski_83}.
Notwithstanding that,
the right-hand side of the preceding equation can be correctly written using the pullback operator,
and this led to the introduction of the operator~$K$  
in a 1986 paper by Batlle \textit{et al}.\
\cite{BGPR_86}:
$$
K \cdot g = 
\FL^*\left(\parder{g}{q^i}\right) v^i +
\FL^*\left(\parder{g}{p_i}\right) \parder{L}{q^i}
\,.
$$
Hence, $K$ was introduced as a kind of differential operator $K \colon \Cinfty(\T^*Q) \to \Cinfty(\T Q)$.
The same authors used it intensively
to develop the Lagrangian and Hamiltonian constraint algorithms, as well as 
to study the relationship between the different types of constraints in singular theories \cite{BGPR_86,BGPR_87,Pon_88}.

A geometrical description of this differential operator was presented in
\cite{CL_87}
by working on the Pontryagin bundle,
but shortly afterwards it was realized that
$K$ is indeed a 
\textsl{vector field along the Legendre map} 
\cite{GP_89}. In this last paper, two different but equivalent constructions of~$K$ 
were presented,
one that is reminiscent of 
the definition of the symplectic gradient (Hamiltonian vector field) on a symplectic manifold,
and another one involving
Tulczyjew's diffeomorphism
$\T^* \T Q \to \T \T^*Q$
\cite{Tul_1976b}.

In addition to writing equations of motion and relating the Lagrangian and Hamiltonian constraint algorithms,
the operator~$K$ has also been used to study Noether's theorem for singular systems,
as well as the relation between the generators of gauge and rigid symmetries both
in the Lagrangian and Hamiltonian formalisms
\cite{GraPon_1994,GraPon_1988,GP_92a,GraPon_2000,GraPon_2001}.
It has also been used in the study of the Hamilton--Jacobi equation \cite{CGMMMR_2009}, 
and for Lagrangian systems whose Legendre map degenerates on a hypersurface
\cite{PV_00}.

The operator~$K$ can be defined for higher-order Lagrangian systems
\cite{BGPR_1988,CL_92,GPR_91},
where it has the same applications as in the first-order case,
see also
\cite{GraPon_1995,GPR_92}. Some work for the case of field theories can be found in
\cite{EMMR_03,RRS_05}.

In any formulation of the Euler--Lagrange equations, 
either 
with the help of the Lagrangian 2-form,
or with the operator~$K$,
or in coordinates,
it is clear that the non-regularity of the equations 
is due to a singular linear factor
(in coordinates the Hessian matrix 
$\parderr{L}{v^i}{v^j}$)
in front of the highest-order derivative.
This led to the consideration of 
\textsl{linearly singular systems} in
\cite{GP_92},
as a particular class of implicit differential equations,
those that can be written in coordinates as
$A(x) \dot x = b(x)$,
with $A(x)$ a singular matrix.

Lagrangian and Hamiltonian systems are essentially conservative.
Nevertheless, 
certain dissipative systems can be described through
G. Herglotz's variational principle 
\cite{Her_30,Her_85},
similar to Hamilton's variational principle,
but with an additional variable $s$ whose time-derivative is the Lagrangian, so that $s$ is essentially the action:
$$
\frac {\d}{\d t} 
\left( \frac {\partial L}{\partial \dot q^i} \right) - 
\frac {\partial L}{\partial q^i}
=
\frac {\partial L}{\partial s} 
\frac {\partial L}{\partial \dot q^i}
\,,
\qquad 
\dot s = L
\,.
$$
The resulting equations are the so-called Herglotz--Euler--Lagrange equations, and they can be expressed in terms of contact geometry
\cite{BH_16,BGG_17,Gei_08}.
In fact, contact geometry is a very suitable tool to describe these kinds of system,
both in the Hamiltonian \cite{Bra_17,BCT_17,BLMP_20,CG_19,LL_19,LS_17,GG_22} and Lagrangian descriptions \cite{CCM_18,LL_19a,GGMRR_20a,MR_18}. 

The aim of this paper is 
to define and study the properties of the time-evolution operator~$K$
for dissipative systems in the context of contact mechanics.
We show how the main features of this operator
in ordinary mechanics extend to the contact case.
In particular, we show that 
it allows us to write the evolution equations even when the Lagrangian is singular,
and we prove that its application to constraint functions on the Hamiltonian formalism yields constraint functions on the Lagrangian one.
We also provide necessary and sufficient conditions for the Lagrangian constraint functions obtained in this way to be $\FL$-projectable.
To obtain these results, 
we extend some of the theory of almost-regular Lagrangian functions to the framework of contact mechanics.

We also introduce two new, equivalent formulations of the contact Hamiltonian equations, 
namely Equations~\eqref{cH_eq_1} and~\eqref{cH_eq_2}. 
These formulations do not depend on the Reeb vector field 
and are thus particularly well-suited 
for cases in which the Reeb field is not well-defined. 
In particular, 
these results allow us to compute the examples presented in Section~\ref{sec:examples}.

The paper is organized as follows. 
In Section~\ref{Kopsymp}, we briefly review the definition and main properties of the evolution operator in the context of ordinary mechanics.
Section~\ref{contact_mechanics} is devoted to contact dynamics. We introduce the main concepts of contact manifolds, and we define contact Hamiltonian systems and the corresponding contact Hamiltonian equations, 
presenting them in several equivalent forms, one of which notably does not involve the Reeb vector field. 
We provide an account of the contact Lagrangian formalism and its associated contact Hamiltonian formulation.
We also present some results for almost-regular contact Lagrangian functions. 
In Section~\ref{sec:contact_K_operator}, we provide an intrinsic characterization of the evolution operator~$K$ for contact systems in two equivalent ways. 
We also analyse the main properties of the evolution operator.
Finally, Section~\ref{sec:examples} contains two illustrative examples:
the simple pendulum and Cawley's Lagrangian, 
both with an additional dissipation term.

Throughout the paper, all the manifolds and maps are smooth.
Sum over crossed repeated indices is understood.

\section{The evolution operator of Lagrangian/Hamiltonian mechanics}
\label{Kopsymp}

The definition of the evolution operator~$K$ 
is based on the concept of vector field along a map.
The usage of vector fields along a path is ubiquitous in differential geometry, 
since the velocity $\gamma'$ of a path~$\gamma$ is a vector field along~$~\gamma$.
Other examples are provided by parallel transport and the Frenet frame in Riemannian geometry.
Less widespread is the notion of vector field along a map,
that is, a lift of a map to the tangent bundle
\cite[section 8.6]{Bou_Var}
\cite[p.\,376]{KMS_93}
or, more generally, that of 
a section along a map
\cite[p.\,36]{Poo_81}
---see
\cite{Car_96}
for a review on their applications. 

In particular, a vector field $Z$ along a map
$F\colon M \to N$
is a map $Z \colon M \to\T N$ such that 
$\tau_N \circ Z = F$,
where $\tau_M \colon \Tan M \to M$ denotes the canonical projection. 
The set of such maps is denoted as $\vf(F)$,
and is a
$\Cinfty(M)$-module. 
Trivial examples of vector fields along~$F$ are 
$Y \circ F$ and $\T F \circ X$,
for vector fields $X \in \vf(M)$ and $Y \in \vf(N)$.
In a similar way we can speak of differential forms along a map.
The inner contraction between 
$Z \in \vf(F)$ and
$\alpha \in \df^k(F)$
can be defined in an obvious way.
See 
\cite{CLM_89}
for more examples of operations with sections along a map.

Let $L\in\Cinfty(\T Q)$ be a Lagrangian function and
$\F L\colon\T Q\to\cT Q$ its
associated Legendre map.
The \textsl{time-evolution operator}~$K$
associated with $L$ is the vector field along the Legendre map, 
$K \colon \T Q \to\T(\cT Q)$, 
satisfying the following conditions:
\begin{enumerate}
\item
(Dynamical condition):
Given the canonical $2$-form $\omega_Q\in\Omega^2(\cT Q)$ and the Lagrangian energy $E_L\in\Cinfty(\T Q)$ defined by $L$, then,
$$ \FL^* (i_K(\omega_Q \circ \FL))= \d E_L $$
\item\label{item_2} 
(Second-order condition): 
If $\pi_Q\colon\cT Q \to Q$ is the canonical projection, then,
$$ \T\pi_Q \circ K = \Id_{\T Q}\,. $$
\end{enumerate}

The existence and uniqueness of this operator is discussed in \cite{GP_89}, 
and it can be shown that it 
can be alternatively defined as 
$K=\chi\circ\d L$, 
where 
$\chi \colon \Tan^*\Tan Q\to \Tan\Tan^*Q$ 
is the canonical isomorphism introduced in 
\cite{Tul_1976b}.
If $(q^i,v^i)$ and $(q^i,p_i)$ are natural local coordinates on $\Tan Q$ and $\Tan^*Q$, respectively,
the local expression of $K$ is
\begin{equation}\label{K.coordinates}
    K(q,v) = v^i\restr{\parder{}{q^i}}{\F L(q,v)}+ \parder{L}{q^i}
    \restr{\parder{}{p_i}}{\F L(q,v)}\,.
\end{equation}

By definition, $\varphi: I\subset\R \to \T Q$ is an integral curve of $K$ if $\T\F L \circ \dot\varphi = K \circ \varphi$. Furthermore, as a consequence of the second-order condition (in the item \ref{item_2}), we have that $\varphi=\dot\phi$, for $\phi: \R \to Q$, that is, $\varphi$ is a holonomic curve. We have the following commutative diagram

\begin{center}
    \begin{tikzcd}[column sep=2cm, row sep=1.75cm]
        & \T(\T Q) \arrow[r, "\T\F L"] \arrow[d, "\tau_{\T Q}"] & \T(\cT Q) \arrow[d, "\tau_{\cT Q}"] \\
        \R \arrow[ur, "\dot\varphi"] \arrow[r, "\varphi"] \arrow[dr, "\phi"] & \T Q \arrow[r, "\F L"] \arrow[d, "\tau_Q"] \arrow[ur, "K"] & \cT Q \\
        & Q
    \end{tikzcd}
\end{center}

\noindent In coordinates, the integral curves of $K$ are the solutions to the equations, 
\begin{equation}
\dot{q} = v\,,\qquad \dot{p} = \parder{L}{q}\,,\qquad \text{and} \qquad p = \parder{L}{v}\,,
\label{K.integral.curves}
\end{equation}
which are obviously equivalent to the Euler--Lagrange equations for $L$.

The Lagrangian and Hamiltonian descriptions of a dynamical Lagrangian system can be unified by means of the operator $K$, as set out in the following properties
\cite{BGPR_86,BGPR_87,CL_87,GP_89}:
\begin{itemize}
\item 
Let $X_L$ be a holonomic vector field on $\T Q$ (that is, a {\sl second-order differential equation} or {\sc SODE}) which is a solution to the Lagrangian equation $i_{X_L}\omega_L=\d E_L$, where $\omega_L\in\Omega^2(\T Q)$ is the Lagrangian $2$-form associated with $L$.
Then $\varphi:\R\to\T Q$ is an
integral curve of $X_L$ if, and only if, it is an integral curve of $K$.

As a direct consequence of this fact, the relation between $K$ and $X_L$ is
\begin{equation}\label{K2a}
\T\F L \circ X_L = K \,.
\end{equation}
\item
Let $X_H$ be a vector field on $\F L(\T Q)\subseteq \cT Q$ which is a solution to the Hamilton equations in the Hamiltonian formalism associated with the Lagrangian system $(\T Q,L)$.  
Then the path $\psi:\R \to \cT Q$ is an integral curve of $X_H$ if, and only if,
\begin{equation}\label{K3}
\dot\psi = K\circ\T\pi_Q\circ\dot\psi \,,
\end{equation}
and, as a consequence, in the final constraint submanifold $S_f$ the relation between $K$ and $X_H$ is
\begin{equation}\label{K4}
    X_H \circ \F L = K \,.
\end{equation}
\end{itemize}
In this way, the equivalence between the
sets of solutions of Euler--Lagrange equations and Hamilton equations is established via the evolution operator.
\begin{itemize}
\item 
The complete classification of Lagrangian and Hamiltonian constraints
appearing in the constraint algorithms for singular dynamical systems
is also achieved using the operator~$K$. 
In fact, all Lagrangian constraints can be obtained from the Hamiltonian
ones using the time-evolution operator since,
if $\xi\in\Cinfty(\cT Q)$ is a Hamiltonian
constraint, then $\Lie_K\xi = i_{K}\d\xi $ is a Lagrangian constraint, and all Lagrangian constraints are recovered in this way.
Additionally, if $\xi$ is a first-class constraint (resp.\ a second-class constraint), then $\Lie_K\xi$ is a dynamical constraint (resp.\ a SODE constraint).
\end{itemize}

\section{Contact dynamics}\label{contact_mechanics}

\subsection{Hamiltonian systems}

A \textsl{contact manifold} is a pair $(M,C)$ such that $M$ is a $(2n+1)$-dimensional manifold and $C$ is a corank-one maximally non-integrable distribution on $M$. We call $C$ a \textsl{contact distribution} on~$M$. Note that $C$ can be locally described, on an open neighbourhood $U$ of each point $x\in M$, as the kernel of a 1-form $\eta\in\Omega^1(U)$ such that $\eta\wedge(\d\eta)^n$ is a volume form on $U$.

A {\sl co-orientable contact manifold} is a pair $(M,\eta)$, where $\eta$ is a 1-form on $M$ such that $(M,\Ker\eta)$ is a contact manifold. Then $\eta$ is called a {\sl  contact form}. Since we are interested in local properties of contact manifolds and related structures, we will hereafter restrict ourselves to co-oriented contact manifolds. To simplify terminology, co-oriented contact manifolds will be called contact manifolds as in the standard modern literature on contact geometry \cite{Bra_17,LL_19,GGMRR_20a}. Thus, $(M,\eta)$ is a contact manifold if $\eta\wedge(\d\eta)^n$ is a volume form.

A contact 1-form $\eta\in\Omega^1(M)$ induces a decomposition of the tangent bundle $\T M$:
$$ \T M = \Ker\eta\oplus\Ker\d\eta\,. $$
Note that if $\eta$ is a contact form on $M$, then $f\eta$ is also a contact form on $M$ for every nowhere-vanishing function $f\in\Cinfty(M)$, and $\Ker \eta = \Ker f\eta$. 

Let $(M,\eta)$ be a contact manifold. There exists a vector bundle isomorphism $B \colon \T M\to\cT M$ given by
\begin{equation}
\label{canis}
 B(v) = i_v(\d\eta)_x + (i_v\eta_x)\eta_x\,,\qquad \forall v\in \T_x M,\quad\forall x\in M\,.     
\end{equation}
This isomorphism can be extended to a $\Cinfty(M)$-module isomorphism $B \colon \X(M)\to\Omega^1(M)$ in the natural way.

Given a contact manifold $(M,\eta)$, there exists a unique vector field $R\in\X(M)$, called the {\sl  Reeb vector field}, such that $i_{R}\d\eta = 0$ and $i_{R}\eta = 1$, or equivalently, $R = B^{-1}(\eta)$. We have that $\Lie_R\eta = 0$ and, therefore, $\Lie_R\d\eta = 0$.

\begin{theorem}[Darboux theorem \cite{AM_78,LM_87}]
    Given a contact manifold $(M,\eta)$ with $\dim M=2n+1$, around every point $x\in M$ there exist local coordinates $(q^i, p_i, s)$, with $i = 1,\dotsc,n$, called {\sl  Darboux coordinates}, such that
    $$ \eta = \d s - p_i\d q^i\,. $$
    In these coordinates, $R = \tparder{}{s}$.
\end{theorem}

\begin{example}[Canonical contact manifold]\label{ex:2.2}
Consider the product manifold $M = \cT Q\times\R$, where $Q$ is an $n$-dimensional manifold. Given a local chart $(q^i)$ on $Q$, there exist canonical coordinates $(q^i,p_i)$ in the cotangent bundle $\cT Q$. In addition, $\R$ has a natural coordinate $s$. This gives rise to a coordinate system $(q^i, p_i, s)$ on $\cT Q\times \R$. Then $\eta_Q = \d s - \theta_Q$, where $\theta_Q$ is the pull-back of the Liouville 1-form $\theta\in\Omega^1(\cT Q)$ relative to the canonical projection $\cT Q\times \R\to\cT Q$, is a contact form on $M$. In canonical coordinates, 
$$ \eta_Q = \d s - p_i\d q^i\,,\qquad R=\parder{}{s}\,.$$
    Note that the coordinates $(q^i,p_i,s)$ are Darboux coordinates on $M$. 
\end{example}

A {\sl contact Hamiltonian system} \cite{BCT_17,LL_19,GGMRR_20a} is a triple $(M,\eta,h)$, where $(M,\eta)$ is a contact manifold and $H\in\Cinfty(M)$ is called a \textsl{Hamiltonian function}. Given a contact Hamiltonian system $(M,\eta,H)$, there exists a unique vector field $X_H\in\X(M)$, called the {\sl  contact Hamiltonian vector field} of $H$, satisfying any of the following equivalent conditions:
\begin{enumerate}[(1)]
    \item $i_{X_H}\d\eta = \d H - (\Lie_R H)\eta$ \quad and \quad $i_{X_H}\eta = -H$,
    \item $\Lie_{X_H}\eta = -(\Lie_R H)\eta$ \quad and \quad $ i_{X_H}\eta = -H$,
    \item $B(X_H) = \d H - (\Lie_R H + H)\eta$. 
\end{enumerate}
A vector field $X\in\X(M)$ is said to be \textsl{Hamiltonian} relative to the contact form $\eta$ if it is the Hamiltonian vector field of a function $H\in\Cinfty(M)$. Unlike in the case of symplectic mechanics, a Hamiltonian function $H$ may not be preserved along the integral curves of its associated contact Hamiltonian vector field $X_H$. More precisely,
$$ \Lie_{X_H}H = -(\Lie_R H)H\,. $$
A function $f\in\Cinfty(M)$ such that $\Lie_{X_H}f = -(\Lie_Rf)f$ is called a \textsl{dissipated quantity} \cite{GGMRR_20a}.
In Darboux coordinates, the contact Hamiltonian vector field $X_H$ reads
\begin{equation}\label{Eq:HamCoor} X_H = \parder{H}{p_i}\parder{}{q^i} - \left( \parder{H}{q^i} + p_i\parder{H}{s} \right)\parder{}{p_i} + \left( p_i\parder{H}{p_i} - H \right)\parder{}{s}\,. \end{equation}
Its integral curves, $\gamma(t) = (q^i(t), p_i(t), s(t))$, satisfy the system of differential equations
$$
    \frac{\d q^i}{\d t} = \parder{H}{p_i}\,,\quad
    \frac{\d p_i}{\d t} = - \left( \parder{H}{q^i} + p_i\parder{H}{s} \right)\,,\quad
    \frac{\d s}{\d t} = p_j\parder{H}{p_j} - H\,,\qquad i = 1,\dotsc,n\,.
$$

\begin{proposition}
    Given a contact Hamiltonian system~$(M,\eta,H)$ and a vector field $X\in\X(M)$, the following statements are equivalent:
    \begin{enumerate}[(1)]
        \item The vector field $X$ is the contact Hamiltonian vector field of $H$.
		\item The vector field $X$ satisfies
		\begin{equation} \label{cH_eq_1}
				(i_{X} \d \eta) \wedge \eta = (\d H) \wedge \eta 
				\qquad \text{and} \qquad
				i_{X} \eta = -H \,.
		\end{equation}
		
		\item The vector field $X$ satisfies
		\begin{equation}\label{cH_eq_2}
				i_{X} (\eta \wedge \d \eta) = \Omega\,,
		\end{equation}
		where $\Omega$ is the 2-form on $M$ defined as $\Omega \coloneq - H \d \eta + \d H \wedge \eta$.
		\end{enumerate}
\end{proposition}
\begin{proof}
    It is clear that~\eqref{cH_eq_1} implies \eqref{cH_eq_2}. The converse follows from the fact that if we take the wedge by $\eta$ on both sides of~\eqref{cH_eq_2}, we obtain
    $$(H+i_X \eta ) \eta \wedge \d \eta =0\,,$$
    which implies that $i_X \eta = -H$, because $\eta \wedge \d \eta \neq 0$ as a consequence of $\eta$ being a contact form. 
    
    In order to prove that~\eqref{cH_eq_2} implies any of the other equivalent forms of the contact Hamiltonian stated before, we just need to contract \eqref{cH_eq_2} with the Reeb vector field. We obtain
    $$i_R i_X (\eta \wedge \d \eta) = -i_X \d\eta=i_R \Omega = (\Lie_R H)\eta - \d H \,.$$
    The converse comes directly from wedging with the 1-form $\eta$ on both sides of the equation $i_{X_H}\d\eta = \d H - (\Lie_R H)\eta$. 
\end{proof}
\begin{remark}
    The reason why these two last equivalent conditions have been explicitly separated from the rest is because they do not use the Reeb vector field. This is particularly important when dealing with \textsl{precontact systems}, where Reeb vector fields may not exist, or be unique.

    In~\cite{LGGMR_23,LL_19a,Lai_22}, singular contact Lagrangian functions and precontact systems were studied, and a constraint algorithm for these types of systems was developed. However, those works consider only the case in which the precontact structure admits a Reeb vector field—an assumption that is not always satisfied. Later, a more general definition of precontact manifold was introduced in~\cite{GG_23}.
\end{remark}

\subsection{Lagrangian systems}

Let $Q$ be an $n$-dimensional manifold and consider the product manifold 
$\T Q\times\R$ equipped with adapted coordinates $(q^i,v^i, s)$ and the canonical projections,
$$ s\colon \T Q\times\R\to\R \,, \qquad \tau_1\colon \T Q\times\R\to\T Q\,, \qquad \tau_0\colon \T Q\times\R\to Q\times\R\,. $$
Note that the projections $\tau_1$ and $\tau_0$ are the projection maps of two vector bundle structures. We will mostly use the latter. In fact, the bundle given by $\tau_0$ is the pull-back of the tangent bundle $\T Q\to Q$ with respect to the map $Q \times \R \to Q$.

The next step is to extend the usual geometric structures of the tangent bundle, namely the tangent structure and the Liouville vector field, to the product $\T Q\times\R$. Note that we can write
$$ \T(\T Q \times \R) = (\T\T Q \times \R) \oplus_{Q\times\R} (\T Q \times \T\R)\,, $$
so any operation acting on vectors tangent to $\T Q$ also acts on vectors tangent to $\T Q \times \R$. In particular, the vertical endomorphism of $\T\T Q$ yields a \textsl{vertical endomorphism} on $\Tan(\Tan Q \times \R)$, ${\cal J} \colon \T (\T Q\times\R) \to \T (\T Q\times\R)$. Analogously, the Liouville vector field $\Delta\in\X(\T Q)$ yields a \textsl{Liouville vector field} $\Delta\in\vf(\T Q\times\R)$. In fact, this is the Liouville vector field of the vector bundle structure given by $\tau_0$. In adapted coordinates, these objects read
$$ {\cal J} = \frac{\partial}{\partial v^i} \otimes \d q^i \,,\qquad \Delta = v^i\, \frac{\partial}{\partial v^i}\,. $$
Given a path ${\bf c} \colon\R \rightarrow Q\times\R$ with ${\bf c} = (\mathbf{c}_1,\mathbf{c}_2)$, its \textsl{prolongation} to $\T Q\times\R$ is the path
$$ {\bf \widetilde c} = (\mathbf{c}_1',\mathbf{c}_2) \colon \R \longrightarrow \T Q \times \R \,, $$
where $\mathbf{c}_1'$ is the velocity of~$\mathbf{c}_1$. The path ${\bf \widetilde c}$ is said to be \textsl{holonomic}. A vector field $\Gamma \in \X(\T Q \times \R)$ is said to satisfy the \textsl{second-order condition} (for short: it is a {\sc sode}) when its integral curves are holonomic, or equivalently, if ${\cal J} \circ \Gamma = \Delta$. In adapted coordinates, if ${\bf c}(t)=(c^i(t), s(t))$, then
$$ {\bf \widetilde c}(t) = \Big( c^i(t),\frac{\d c^i}{\d t}(t), s(t) \Big) \,. $$
The local expression of a {\sc sode} is
\begin{equation}
\label{localsode2}
\Gamma= 
v^i \frac{\partial}{\partial q^i} +
f^i \frac{\partial}{\partial v^i} + 
g\,\frac{\partial}{\partial s}
\,.
\end{equation}

\begin{definition}
\label{dfn:lagrangian-function}
A {\sl Lagrangian function} 
is a function $L \colon\T Q\times\R\to\R$. The {\sl Lagrangian energy}
associated with~$L$ is the function $E_L := \Delta(L)-L\in\Cinfty(\T Q\times\R)$. The {\sl Cartan forms}
associated with $L$ are defined as
\begin{equation}
\label{eq:thetaL}
\theta_L = 
{}^t{\cal J} \circ \d L 
\in \Omega^1(\T Q\times\R) 
\,,\quad
\omega_L = 
-\d \theta_L
\in \Omega^2(\T Q\times\R) 
\,.
\end{equation}
The {\sl Lagrangian 1-form} is
$$
\eta_L=\d s-\theta_L\in\Omega^1(\T Q\times\R)
\,,
$$
and satisfies that $\d\eta_L=\omega_L$.
\\
The pair $(\T Q\times\R,L)$ is a {\sl Lagrangian system}.
\end{definition}

If we take natural coordinates $(q^i, v^i, s)$ in $\T Q\times\R$, 
the Lagrangian 1-form $\eta_L$ is written as
\begin{equation}
\label{eq:etaL}
\eta_L = \d s - \frac{\partial L}{\partial v^i} \,\d q^i \,,
\end{equation}
and, consequently,
\begin{equation*}
\d\eta_L = 
-\frac{\partial^2L}{\partial s\partial v^i}\d s\wedge\d q^i 
-\frac{\partial^2L}{\partial q^j\partial v^i}\d q^j\wedge\d q^i 
-\frac{\partial^2L}{\partial v^j\partial v^i}\d v^j\wedge\d q^i\,.
\end{equation*}

The next structure to be defined is the Legendre map. This map is studied in more detail in Section~\ref{sec:almost_regular}.

\begin{definition}
Given a Lagrangian 
$L \colon \T Q\times\R \to \R$, 
its {\sl Legendre map}
is the fibre derivative of~$L$,
considered as a function on the vector bundle
$\tau_0 \colon \T Q\times\R \to Q \times \R$;
that is, the map
$\F L \colon \T Q \times\R \to \cT Q \times \R$ 
given by
$$
    \F L (v_q,s) = \left( \F L(\cdot,s) (v_q),s \right)\,,
$$
where $L(\cdot,s)$ is the Lagrangian with $s$ frozen.
\end{definition}

\begin{proposition}\label{Prop-regLag}
For a Lagrangian function, $L$ the following conditions are equivalent:
\begin{enumerate}[(1)]
\item
The Legendre map
$\F L$ is a local diffeomorphism.
\item The fibre Hessian $\F^2L \colon
\T Q\times\R \longrightarrow (\cT Q\times\R)\otimes_{Q\times\R} (\cT Q\times\R)$ of $L$ is everywhere non-degenerate.
\item The pair $(\T Q\times\R,\eta_L)$ is a contact manifold.
\end{enumerate}
\end{proposition}

The proof of this result follows from the expressions in natural coordinates, where
$$
{\cal F}L:(q^i,v^i, s)  \longrightarrow  \left(q^i,
\frac{\displaystyle\partial L}{\displaystyle\partial v^i}, s
\right)\,,
$$
$$
{\cal F}^2 L(q^i,v^i,s) = (q^i,W_{ij},s)
\,,\quad
\hbox{with}\ 
W_{ij} = 
\left( \frac{\partial^2L}{\partial v^i\partial v^j}\right)
\,.
$$
The conditions in the proposition are also equivalent to the matrix 
$W= (W_{ij})$ 
being everywhere non-singular.

\begin{definition}
A Lagrangian function $L$ is said to be {\sl regular} if the equivalent
conditions in Proposition \ref{Prop-regLag} hold.
Otherwise, $L$ is called a {\sl singular} Lagrangian.
In particular, 
$L$ is said to be {\sl hyperregular} 
if $\F L$ is a global diffeomorphism.
\end{definition}

\begin{remark}
As a result of the preceding definitions and results,
every {\sl regular} Lagrangian system 
has associated the contact Hamiltonian system
$(\T Q\times\R, \eta_L, E_L)$.
\end{remark}

Given a regular Lagrangian system $(\T Q\times\R,L)$, the {\sl Reeb vector field} $R_L\in\X(\T Q\times\R)$ is uniquely determined by the conditions
\begin{equation}\label{eq:reeb}
    i_{R_L}\d\eta_L=0\,,\qquad i_{R_L}\eta_L=1 \,.
\end{equation}
Its local expression is
\begin{equation}
\label{coorReeb}
R_L=\frac{\partial}{\partial s}-W^{ji}\frac{\partial^2L}{\partial s \partial v^j}\,\frac{\partial}{\partial v^i} \,,
\end{equation}
where $(W^{ij})$ is the inverse of the partial Hessian matrix,
namely 
$W^{ij} W_{jk} = \delta^i_{k}$.

Note that the Reeb vector field does not appear
in the simplest form $\partial/\partial s$.
This is due to the fact that the natural coordinates in $\T Q \times \R$
are not Darboux coordinates for $\eta_L$.

For a regular Lagrangian, there exists a unique \textsl{Lagrangian vector field} such that
\begin{equation}\label{contact_EL}
        i_{X_L} \d \eta_L = \d E_L -( \Lie_{R_L} E_L ) \eta_L \,,\qquad
    i_{X_L} \eta_L = - E_L \,.
\end{equation}

\begin{proposition}
    If the Lagrangian $L$ is regular, then $X_L$ is a second-order differential vector field and its integral curves satisfy the so-called Herglotz--Euler--Lagrange equations. 

    In local coordinates, this means that 
    $\displaystyle X_L = v^i \frac{\partial}{\partial q^i} + f^i \frac{\partial}{\partial v^i} + g \frac{\partial}{\partial s}$, with 
    \begin{align}
        g&=L\,,
        \\
        \parderr{L}{v^j}{v^i}f^j + \parderr{L}{q^j}{v^i}v^j +&
        \parderr{L}{s}{v^i} L -\parder{L}{q^i}  = \parder{L}{s}\parder{L}{v^i}
    \end{align}
\end{proposition}

Note that the Herglotz--Euler--Lagrange equations, since they come from a variational principle, exist regardless of the regularity of the Lagrangian. As a matter of fact, we can write these equations in a geometric way even for singular Lagrangians. In general, the solutions to the Herglotz--Euler--Lagrange equations can be seen as the integral curves of the second-order vector fields $X_L$ satisfying the equations
\begin{equation}
        i_{X_L} \d \eta_L = \d E_L +\left(\parder{L}{s}\right) \eta_L \,,\qquad
    i_{X_L} \eta_L = - E_L \,.
\end{equation}
These are well-defined since $\tparder{L}{s}$ is defined canonically in $\Tan Q \times \R$, and they are equivalent to \eqref{contact_EL}, as one can check that, in the regular case, $\Lie_{R_L} E_L = -\tparder{L}{s}$. If the Lagrangian is singular, these equations do not directly imply that the vector field is second-order, and the condition must be added. Also, in the singular case, solutions may not exist at every point, and if they do, they might not be unique. This requires a more careful study of both the Lagrangian and Hamiltonian formalisms in the singular case.

\subsection{Hamiltonian formalism for contact Lagrangian systems}

Let us analyse how to construct a contact Hamiltonian formalism that is related to the contact Lagrangian formalism we presented before. That is, we describe a contact Hamiltonian system on $\Tan^*Q \times \R$ which is in correspondence, via the Legendre map, with the Lagrangian system $(\Tan Q \times \R, L)$.

Recall, from Example \ref{ex:2.2}, that $\Tan^* Q \times \R$ is endowed with a natural contact form $\eta_Q = \d s - \theta_Q$.
A direct computation shows that
$$ \eta_L = \FL ^* (\eta_Q)\,. $$
If the Lagrangian is regular, then the energy $E_L$ is, at least locally, $\FL$-projectable (since $\FL$ is a local diffeomorphism). Namely, we can (locally) always find a \textsl{contact Hamiltonian function} $H\colon \Tan ^* Q \times \R \to \R$ such that
$E_L = \FL^*(H) = H\circ\FL$.
In the hyperregular case, the energy is globally $\FL$-projectable and the system $(\Tan^* Q\times \R , \eta_Q, H)$ satisfies
$$\FL_* (R_L) = R\,, \qquad \text{and} \qquad \FL_* (X_L)=X_H\,,$$
for $R,X_H \in \X (\Tan^* Q \times \R)$, respectively, the Reeb vector field of $\eta_Q$ and the contact Hamiltonian vector field associated with the system. This means that the two formalisms are $\FL$-related, and the map $\xi \mapsto \FL \circ \widetilde \xi$, where $\xi \colon I \to Q \times \R$ is a solution to the Herglotz--Euler--Lagrange equations, sends solutions to solutions. Conversely, if $\psi \colon I \to \Tan^* Q \times \R$ is a integral path of $X_H$, then the map $\psi \mapsto \pi_0 \circ \psi$, with $\pi_0 \colon \Tan^* Q \times \R \to Q\times \R$ the canonical projection, also sends solutions of the contact Hamiltonian formalism to solutions of the Herglotz--Euler--Lagrange equations.

Let us point out that, regardless of the regularity of the Lagrangian, any solution on $\Tan^* Q \times \R$ which is $\FL$-related to a solution to the Herglotz--Euler--Lagrange equations fulfils the equations
\begin{equation}\label{Herglotz--Dirac_eq}
    \begin{dcases}
    p= \parder{L}{v}\,,
    \\
    \dot p = \parder{L}{s}\parder{L}{v} + \parder{L}{q}\,,
    \\
        \dot s = L\,,
    \end{dcases}
\end{equation}
called the \textsl{Herglotz--Dirac} equations.

\subsection{Almost-regular Lagrangian functions}
\label{sec:almost_regular}
In this section we present some results on singular Lagrangian functions. Namely, we extend the results presented on~\cite{GraPon_2001} to the case of contact mechanics. We refer the reader to the aforementioned article for a detailed presentation of the 
results and their proofs.

Let us start by stating some preliminary results about fibre derivatives.
Recall that,
given a (not necessarily linear) bundle map $f \colon E \to F$
between two vector bundles over a manifold~$B$,
the fibre derivative of~$f$ is the map
$\mathcal{F}f \colon 
E \to \mathrm{Hom}(E,F) \approx F \otimes E^*$
obtained by
restricting $f$ to the fibres, 
$f_b \colon E_b \to F_b$,
and computing the usual derivative
of a map between vector spaces:
$\mathcal{F}f(e_b) = \mathrm{D} f_b(e_b)$
(see \cite{Gra_00} for a detailed account).
This applies in particular 
when the second vector bundle is trivial of rank~1,
that is,
for a function
$f \colon E \to \R$;
then
$\mathcal{F}f \colon E \to E^*$.
This map also has a fibre derivative
$\mathcal{F}^2 f \colon E \to E^* \otimes E^*$,
which can be called the fibre hessian of $f$:
for every $e_b \in E$,
$\mathcal{F}^2 f(e_b)$ can be considered as a
symmetric bilinear form on $E_b$.
One can check that
$\mathcal{F}f$ is a local diffeomorphism 
at $e \in E$
if, and only if, $\F^2f(e)$ is non-degenerate.
Indeed, it is a direct consequence of the fact that $\Ker \Tan (\F f) \subset V(E)$ and
$$v_x \in \Ker \F ^2 f(e_x) \quad \Longleftrightarrow \quad \mathrm{vl}_{E}(e_x,v_x) \in \Ker \T_{e_x} (\F f)\,,$$
which can be easily seen with the local expressions of the objects involved.

Another interesting result involving the fibre derivative that will be useful later is the following. Let $\xi \colon E \to E $ be a bundle map with associated vertical field $X= \xi ^v = \mathrm{vl}_E \circ \xi$ on $E$, and let $f \colon E \to \R$ be a function, then
$$X \cdot f = \langle\F f, \xi \rangle\,.$$
In particular, if we take $\xi = \mathrm{Id}_E$, we have
\begin{equation}\label{Liouville_prop_1}
	(\Delta_E \cdot f) (e_x) = \langle \F f(e_x) , e_x \rangle\,,
\end{equation}
where $\Delta_E = \mathrm{vl}_E(\Id_E)$ is the Liouville vector field of the vector bundle $E$. By applying the chain rule to this last expression, we also obtain 
\begin{equation}\label{Liouville_prop_2}
	\F (\Delta_E \cdot f)(e_x) = \F f (e_x)+\F^2 f(e_x)\cdot e_x\,.
\end{equation}

In the case of contact Lagrangian systems, our framework involves two real vector bundles over the same base, namely 
$  \tau_0\colon\Tan Q \times \mathbb{R} \to Q \times \mathbb{R}
$ and $
 \pi_0 \colon \Tan^*Q \times \mathbb{R} \to Q \times \mathbb{R}.$
Taking this into account, one can readily see that most of the theory of singular Lagrangian systems extends to contact mechanics. In particular, the fibre derivative of a Lagrangian function $L\colon\T Q\times\R\to\cT Q\times\R$ is its Legendre map.

As in the case of ordinary
mechanics, to obtain some general results in the singular case, we still need to impose some weak regularity conditions on the Lagrangian functions.
Namely, following \cite{GN_79}, we say that a Lagrangian function $L\in \Cinfty(\T Q \times \R)$ is \textsl{almost-regular} if: 
(i)   the image of the Legendre map $P_0 \coloneqq \FL(\Tan Q\times\R) \overset{j}{\hookrightarrow} \Tan^*Q\times \R$ is a closed submanifold, called the \textsl{primary Hamiltonian constraint submanifold}, 
(ii) the induced map $\FL_0 \colon \Tan Q\times \R\to P_0$ is a submersion and has connected fibres.
From a local point of view, it suffices to assume that $\FL$ has constant rank.

Since $P_0 \coloneqq \FL (\Tan Q\times \R)$ is a closed submanifold, it is locally defined by the vanishing of an independent set of functions $\{\phi_{\mu}\}_{\mu=1,\dots,m}$, with linearly independent differentials $\d \phi_{\mu}$ at every point. We call these functions \textsl{primary Hamiltonian constraint functions}. 
Now, we recall the following two lemmas:
\begin{lemma}\label{lemma_ker_1}
    Let $\alpha \colon N \to P $ be a submersion at $y\in N$. Then, there exists an open neighbourhood $V\subset N$ of $y$ such that $\alpha(V)\subset P$ is a submanifold. Additionally, we have:
    \begin{enumerate}
        \item $\mathrm{Im} \Tan_y \alpha = \Tan_{\alpha(y)}(\alpha(V))$\,.
        \item If $\alpha(V) \subset P$ is the submanifold defined locally by the vanishing of a set of functions $\phi_{\mu}$, with linearly independent differentials $\d \phi_{\mu}$ at every point, then $\Ker ^t \Tan_y\alpha$ admits $\d_{\alpha(y)} \phi_{\mu}$ as a basis, and $\Ker ^t \Tan\alpha$ admits $\d\phi_{\mu}\circ \alpha$ as a frame over $V$.
    \end{enumerate}
\end{lemma}

\begin{lemma} \label{lemma_ker_2}
    Let $P_0 \overset{j}{\hookrightarrow} P$ be a submanifold defined by the vanishing of $\{\phi_{\mu}\}_{\mu=1,\dots,m}$, such that their differentials $\d \phi_{\mu}$ are linearly independent at every point of $P_0$. Then, $\Ker ^t\Tan (j)$ admits $\{\d \phi_{\mu} |_{P_0}\}_{\mu=1,\dots,r}$ as a frame. Also, at every point $x\in P_0$ a tangent vector $v_x \in \Tan_x P$ is in $\Tan_x P_0$ if, and only if, for every $\phi_{\mu}$, it satisfies $\d\phi_{\mu}(v_x)=0$.
\end{lemma}

We can apply Lemmas \ref{lemma_ker_1} and \ref{lemma_ker_2} to show that the $\d \phi_{\mu} \circ \FL$ form a (local) frame of the vector subbundle defined by $\Ker ^t\Tan (\FL) \subset \Tan^* (\Tan^*Q\times \R)$.

For every given function $g \in \Cinfty(\Tan^*Q \times \R)$, we can define the map
$$\gamma_g = \F g \circ \F L \colon \Tan Q \times \R \to \Tan Q \times \R\,,$$
where we denote by $\FL$ the fibre derivative of $L$, considered as a function on
$\tau_0 \colon \T Q\times\R \to Q \times \R$,
and $\F g$ the fibre derivative of $g$,
considered as a function on the vector bundle
$\pi_0 \colon \T^* Q\times\R \to Q \times \R$.

If we apply the natural bijection between fibre bundle maps and vertical vector fields; that is, the vertical lift $\mathrm{vl}_{\Tan Q \times \R }$, we obtain the vertical vector field
$$\Gamma_g \coloneqq \gamma_g^{v}=\mathrm{vl}_{\Tan Q \times \R}(\Id_{\Tan Q\times \R}, \F g \circ \FL)\,.$$
In local coordinates, these objects have the following expressions:
\begin{equation}
    \gamma_g \colon (q,v,s) \mapsto \left( q, \parder{g}{p} (\FL (q,v,s)),s \right)\,, \qquad \Gamma_g= \FL^{*}\left( \parder{g}{p}  \right)\parder{}{v}\,. 
\end{equation}
With these objects, we can prove the following result:
\begin{proposition}\label{Ker_TFL}
    The vector fields $\Gamma_{\mu} = \gamma_{\phi_{\mu}}^v$, constructed from the primary Hamiltonian constraint functions $\phi_{\mu}$, form a local frame for $\Ker \Tan(\FL)$. Their local expression is
    $$\Gamma_{\mu}=\gamma_{\mu}^i \parder{}{v^i}\,,$$
    where the functions
    $$\gamma_{\mu}^i= \FL^*\left(  \parder{\phi_{\mu}}{p_i} \right)$$
    form a basis of the kernel of the Hessian matrix $W=\displaystyle
\left( \frac{\partial^2 L}{\partial v^i \partial v^j} \right)$.
\end{proposition}

This last result follows directly from applying the chain rule, since
$$\F (g \circ \F L) = \F^2 L \bullet \gamma_g\,, $$
where the symbol $\bullet$ denotes the composition of the images of both maps. This implies that if $g$ vanishes on the image $\FL (\Tan Q \times \R) \subset \Tan^* Q \times \R$, then $\gamma_g$ is in the kernel of $\F^2 L$. Since, by definition, we take the functions $\phi_{\mu}$ to be linearly independent, the result follows immediately.
The result can also be easily proved using the local coordinate expressions.

Now we give the most important result for the characterization of a Hamiltonian formalism in this context~\cite{BGPR_86,GN_79}.
\begin{proposition}
	If the Legendre map $\FL$ is a submersion, then the Lagrangian energy function $E_L$ is locally projectable. That is, locally, there exists a function $H$ such that $E_L= H\circ\FL$.
\end{proposition}
Again, this proposition can be proved using the local expressions of the objects involved. We can also use the properties of the Liouville vector field~\eqref{Liouville_prop_1}, \eqref{Liouville_prop_2}, as they imply
\begin{equation}
E_L(v_q,s)=(\Delta_{\Tan Q \times \R}\cdot L)(v_q,s)-L(v_q,s) = \langle  \FL (v_q,s),(v_q,s) \rangle-L(v_q,s)\,, 
\end{equation}
\begin{equation} \label{FEL_eq}
	\F E_L (v_q,s)= \F^2L(v_q,s)\cdot (v_q,s)\,.
\end{equation}
Hence, one has
$$\Gamma_{\mu}\cdot E_L = \langle\F E_L ,\gamma_{\mu}\rangle = 0\,,$$
which is equivalent to the Lagrangian energy being projectable via $\FL$, by Proposition~\ref{Ker_TFL}.

Thus, if the Lagrangian function $L$ is almost-regular, then $E_L$ is \textsl{globally} $\FL$-projectable at $P_0$, that is, there exists a \textsl{unique} function $H_0\colon P_0 \to \R$, the \textsl{Hamiltonian function}, such that $\FL^*(H_0)=E_L$. 
Also, as $P_0$ is assumed to be closed, the function $H_0$ can be extended (non-uniquely) to a function $H$ defined on $\Tan^*Q \times \R$.
Note that, for a local study, it is enough to just assume that the Legendre map $\FL$ is a submersion.

Finally, let us prove a result which relates the vector fields $\Gamma_H$ and $\Gamma_\mu$, defined, respectively, from a choice of Hamiltonian function and of primary Hamiltonian constraints, with the Liouville vector field. 
\begin{proposition}\label{prop_resolution_identity}
    Given an almost regular Lagrangian $L$, the choice of a Hamiltonian function $H$ and a set of primary Hamiltonian constraints $\phi_{\mu}$ allows us to write, locally, the identity map of $\Tan Q \times \R$ as
    \begin{equation} \label{resolution_identity}
        \Id _{\Tan Q \times \R} = \gamma_H +  \sum_{\mu}\lambda^{\mu} \gamma_{\phi_{\mu}}\,,
    \end{equation}
    where the $\lambda^{\mu} \in \Cinfty (\Tan Q\times \R)$ are uniquely determined functions.
\end{proposition}
The proof of this result comes from applying the chain rule to the definition of $H$, i.e.\ $E_L = H\circ \FL$, which yields
$$\F E_L (v_q,s) = \F^2 L (v_q ,s )\cdot\gamma_H(v_q,s)\,. $$
Applying \eqref{FEL_eq}, we obtain
$$\F^2 L (v_q ,s )\cdot((v_q,s)-\gamma_H(v_q,s))=0\,,$$
which proves that $(v_q,s)-\gamma_H(v_q,s) = \sum_{\mu} \lambda^{\mu}(v_q,s)\gamma_\mu(v_q,s)$, as we wanted to see.

If we apply the natural bijection between fibre bundle maps and vertical vector fields to~\eqref{resolution_identity}, we obtain the following formula for the Liouville vector field
\begin{equation}\label{resolution_liouville}
    \Delta_{\Tan Q \times \R} = \Gamma_H + \sum_{\mu} \lambda^\mu \Gamma_{\mu}\,.
\end{equation}
Note that, in coordinates, each component of \eqref{resolution_liouville} yields
$$v^i = \gamma_H ^i + \sum_{\mu}\lambda^\mu \gamma_{{\mu}}^i\,,$$
where, in local coordinates, $\gamma_H^i =\FL^*\left(\parder{H}{p_i}\right) $ and  $\gamma_{\mu}^i = \FL^*\left(\parder{\phi_\mu}{p_i}\right)$.
Applying the vector fields $\Gamma_{\nu} = \gamma_{\nu}^i \parder{}{v^i}$ to this last expression, we obtain
$$\gamma_{\nu}^i = \sum_{\mu=1}^{m} (\Gamma_{\nu}\cdot\lambda^{\mu})\gamma_{\mu}^i\,,$$
and, by the linear independency of the $\gamma_{\mu}$, this implies that
\begin{equation} \label{Gamma_vf_condition}
	\Gamma_{\nu} \cdot \lambda^{\mu} = \delta_{\nu}^{\mu}   \,.
\end{equation}
In particular, this shows that the functions $\lambda^
\mu$ are not $\F L$-projectable. As we will show later in the examples, these functions correspond to the velocities which cannot be written in terms of momenta via the Legendre map.

\section{The evolution operator in contact mechanics}
\label{sec:contact_K_operator}

In this section, we present two equivalent intrinsic characterizations of the evolution operator~$K$ for Lagrangian systems in the context of dissipative mechanics, along with some of its most relevant properties.

\subsection{Construction of the evolution operator \texorpdfstring{$K$}{}}

To define the evolution operator in the setting of contact mechanics, we follow the same structure we presented in Section~\ref{Kopsymp}.
We define the evolution operator as the unique vector field along the Legendre map that satisfies three key conditions: the second-order condition and two dynamical conditions.

The condition of $K$ being a vector field along the Legendre map is equivalent to the diagram
\[\begin{tikzcd}[column sep=large, row sep=large]
	& {\Tan (\Tan ^*Q \times \R)} \\
	{\Tan Q \times \R} & {\Tan^*Q \times \R}
	\arrow["{\tau_{\Tan^*Q \times \R}}", from=1-2, to=2-2]
	\arrow["K", from=2-1, to=1-2]
	\arrow["{\FL }"', from=2-1, to=2-2]
\end{tikzcd}\]
being commutative. In other words, the map $K\colon \Tan Q \times \R \to \Tan (\Tan^*Q \times \R)$ satisfies 
\begin{equation}
    \tau_{\Tan^*Q \times \R} \circ K = \FL \,.
\end{equation}
This equation is sometimes referred to as the structural equation, and it implies that, in local coordinates, $K$ is expressed as
\begin{equation}
    K(q,v,s)=\left( q, \parder{L}{v},s\,;\, a(q,v,s), b(q,v,s),c(q,v,s) \right) \,,
\end{equation}
or, alternatively, as
\begin{equation}
	K(q,v)=a^i(q,v,s) \frac{\partial}{\partial q^i}\bigg|_{\FL(q,v,s)} + 
	b_i(q,v,s) \frac{\partial}{\partial p_i}\bigg|_{\FL(q,v,s)} +  c(q,v,s) \frac{\partial}{\partial s}\bigg|_{\FL(q,v,s)}\,,
\end{equation}
where $a^i, b_i$ and $c$ are functions yet to be determined.

Now, the so-called \textsl{second-order condition} for the time-evolution operator $K$ is given by
\begin{equation} \label{contact_second_order_condition}
    \Tan \pi_0^1 \circ K = \mathrm{Id}_{\Tan Q} \,, 
\end{equation}
where $\pi_0^1 \colon \Tan^*Q \times \R \to Q$ is the canonical projection. Writing this in coordinates, we find that 
\begin{equation}
    \Tan\pi_0^1 (q,p,s,v,u,z) = (q,v) \,,
\end{equation}
and so, 
\begin{equation}
    \Tan \pi_0^1 \circ K = (q,a^i) \,.
\end{equation}
Thus, Equation \eqref{contact_second_order_condition} determines that, in coordinates, we have $a^i = v^i$.
Note that this condition is equivalent to $\rho\circ \Tan \pi_0 \circ K = \Id_{\Tan Q \times \R}$, where $\pi_0 \colon \Tan^*Q \times \R \to Q\times \R$ and $\rho \colon \T (Q \times \R) \to \Tan Q \times \R$ are the natural projections.

The last step to fully characterize the time-evolution operator $K$ is via the so-called \textsl{dynamical conditions}. They are 
\begin{align}
\begin{dcases}
        \FL^*(i_K(\d \eta_Q \circ \FL)) = \d E_L + \parder{L}{s} \eta_L\,,
        \\
    i_K(\eta_Q\circ \FL) = -E_L\,,
\end{dcases}
\end{align}
where $\eta_Q$ is the canonical contact form on $\Tan^* Q \times \R$.

In coordinates, these conditions define the functions $b_i$ and $c$. Let us see this, by writing in local coordinates all the expressions involved. First, we have 
$$i_K(\d \eta_Q \circ \FL) = -b_i \d q^i\big|_{\FL(q,v,s)} + v^i \d p_i\big|_{\FL(q,v,s)} \,,$$
and therefore,
$$\FL^*(i_K(\d \eta_Q \circ \FL)) = 
\left( v^j \parderr{L}{v^j}{q^i} - b_i \right)\d q^i
+ v^i\parderr{L}{v^i}{v^j} \d v^i 
+v^i\parderr{L}{v^i}{s} \d s \,.
$$
On the right-hand side, we have 
\begin{align}
    \d E_L + \parder{L}{s} \eta_L = \left( v^j\parderr{L}{v^j}{q^i} - \parder{L}{q^i}- \parder{L}{s}\parder{L}{v^i} \right) \d q^i + v^i\parderr{L}{v^i}{v^j} \d v^i + v^i \parderr{L}{v^i}{s} \d s  \,,
\end{align}
$$$$
and so, equating both expressions, we obtain
\begin{equation}
    b_i = \parder{L}{q^i} + \parder{L}{s} \parder{L}{v^i} \,.
\end{equation}
Now, as
$$i_K(\eta_Q\circ\FL) = c(q,v,s)-v^i\parder{L}{v^i} \,,$$
the second dynamical condition
shows that $c = L$. 

Thus, we have proved that with these three conditions and the fact that the operator $K$ is a vector field along the Legendre map, the evolution operator $K\colon \Tan Q \times \R \to \Tan(\Tan^* Q \times \R)$ is completely determined. Also, we have seen that its coordinate expression is
\begin{equation}
    K(q,v,s) =
    v^i \left(\parder{}{q^i}\circ\FL\right) 
    +\left( \parder{L}{q^i} + \parder{L}{s} \parder{L}{v^i} \right) \left(\parder{}{p_i}\circ\FL\right) 
    + L \, \left(\parder{}{s}\circ\FL\right)\,.
\end{equation}
This vector field along the Legendre map defines an operator 
$K\colon \Cinfty (\Tan^*Q \times\R) \to \Cinfty(\Tan Q \times \R)$
that takes functions defined in the phase space and gives their time derivative in the velocity space. In local coordinates, this derivation is given by
\begin{equation}
        (K\cdot f)(q,v,s) =
    v^i \FL^*\left(\parder{f}{q^i}\right) 
    +\left( \parder{L}{q^i} + \parder{L}{s} \parder{L}{v^i} \right) \FL^*\left(\parder{f}{p_i}\right) 
    + L \, \FL^*\left(\parder{f}{s}\right)\,.
\end{equation}
Note that we can equivalently write $K\cdot f \in \Cinfty(\Tan Q \times \R)$ as the function defined by 
\begin{equation} \label{contact_K_derivation}
    (K\cdot f) (q,v,s) = i_K(\d f\circ \FL)= \langle \d f (\FL (q,v,s)), K(q,v,s) \rangle \,.
\end{equation}

The evolution operator can also be characterized with only one dynamical condition, using the isomorphism $B\colon \Tan (\Tan ^* Q \times \R) \to \Tan^* (\Tan^* Q \times \R)$ defined by the canonical contact $1$-form on $\Tan^*Q\times \R$ (see Equation \eqref{canis}). 

\begin{proposition}
    The evolution operator $K$ can be equivalently characterized as the unique vector field along the Legendre map such that
    \begin{equation}
        \Tan \pi_0^1 \circ K = \mathrm{Id}_{\Tan Q} \,,
    \end{equation}
    and
    \begin{equation} \label{K_formula_def}
        \FL^*(B \circ K) = \d E_L + \left( \parder{L}{s} - E_L\right) \eta_L\,.
    \end{equation}
\end{proposition}
\begin{proof}
    We can prove this in coordinates in a similar way as above.
    The first equation is the second-order condition, and as we saw above, it implies~$a^i=v^i$.
    For the second equation, on the left-hand side we have
    \begin{align}
    \FL^*(B\circ K ) &= \,^t\Tan \FL \circ (i_K(\d\eta_Q\circ \FL) + (i_K(\eta_Q \circ \FL))(\eta_Q\circ \FL))=\\
    &\ \quad\left( v^j \parderr{L}{v^j}{q^i} - b_i - c\parder{L}{v^i}+ v^j\parder{L}{v^j}\parder{L}{v^i}\right)\d q^i\\
    &\ \quad + v^i\parderr{L}{v^i}{v^j}  \d v^i  
+\left( v^i \left( \parderr{L}{v^i}{s} - \parder{L}{v^i} \right) +c\right) \d s \,,
    \end{align}
    and on the right-hand side,
    \begin{align}
            \d E_L + \left(\parder{L}{s} - E_L\right) \eta_L &= 
            \left( v^j\parderr{L}{v^j}{q^i} - \parder{L}{q^i}- \parder{L}{s}\parder{L}{v^i} + v^j\parder{L}{v^j}\parder{L}{v^i} -L\parder{L}{v^i} \right) \d q^i\\ &\ \quad + v^i\parderr{L}{v^i}{v^j} \d v^i +\left( v^i \left( \parderr{L}{v^i}{s} - \parder{L}{v^i} \right)  + L \right) \d s  \,,
    \end{align}
    equating both sides directly yields the local expression of the evolution operator.
\end{proof}

We can also give a construction of the evolution operator $K$ in terms of the so-called \textsl{Tulczyjew's triples}. 
We can use the canonical diffeomorphism
$\chi\colon \Tan^*\Tan Q \longrightarrow \Tan\Tan^* Q$ (see \cite{Tul_1976b})
to define a diffeomorphism between 
$\Tan^* (\Tan Q \times \R)$ and 
$\T (\T^* Q \times \R)$, 
using the natural identification of 
$\Tan (M \times M')$ with $\T M \times \T M'$ 
and of both 
$\Tan \R$ and $\Tan^* \R$ with $\R \times \R$ 
(using the fixed canonical coordinate $s$). 
In coordinates, it reads
\begin{align}
    \widetilde \chi\colon \Tan^*(\Tan Q\times \R) &\longrightarrow \Tan(\Tan^* Q\times \R)\,,
    \\
    (q,v,s,u,p,z)&\longmapsto (q,p,s,v,u,z) \,, 
\end{align}
and then
$$K = \widetilde\chi\circ \left(\d L + L\d s -\parder{L}{s}\eta_L\right)\,.$$

\subsection{The evolution operator and the equations of motion}

This section studies how the Lagrangian and Hamiltonian formalisms can be recovered from the time-evolution operator.

\begin{proposition}
    Let $\xi \colon I\to\Tan Q \times \R$ be a path, and $\dot\xi \colon I \to \Tan(\Tan Q\times \R)$ its canonical lift. Then, $\xi$ is a solution to the Herglotz--Euler--Lagrange equations for a given Lagrangian $L$ if, and only if,
    \begin{equation} \label{contact_K_XL_relation}
        \Tan(\FL) \circ \dot\xi = K\circ \xi\,.
    \end{equation}
\end{proposition}
\begin{proof}
    It is enough to see this in local canonical coordinates. Assume that the path is given by $\xi = (q,v,s)$. 
    Its canonical lift is
    $\dot\xi = (q,v,s ; \dot q, \dot v, \dot s)$.
    Thus, we have
    \begin{equation}
        \Tan(\FL) \circ \dot \xi = \left( q, \parder{L}{v} , s\,;\, \dot q, \dot q \parderr{L}{v}{q} + \dot v \parderr{L}{v}{v} + \dot s\parderr{L}{v}{s}, \dot s \right) \,,
    \end{equation}
    and
    \begin{equation}
        K\circ \xi = \left( q,\parder{L}{v},s\,;\, v, \parder{L}{q} + \parder{L}{s} \parder{L}{v}, L  \right) \,.
    \end{equation}
    Equating both expressions,
    \begin{equation}
        \begin{dcases}
            \dot q = v \,,
            \\
            v \parderr{L}{v}{q} + \dot v \parderr{L}{v}{v} + \dot s\parderr{L}{v}{s} = \parder{L}{q} + \parder{L}{s} \parder{L}{v} \,,
            \\
            \dot s = L \,,
        \end{dcases}
    \end{equation}
    which are precisely the second-order condition and the Herglotz--Euler--Lagrange equations.
\end{proof}
Note that, if a path $\xi \colon I\to \Tan Q\times \R$ satisfies
$\Tan(\FL) \circ \dot \xi = K\circ \xi$
then, necessarily, it can be obtained as the prolongation of a path $\zeta \colon I \to Q\times \R$. Therefore, the equations of motion defined by the evolution operator $K$ incorporate the second-order condition, regardless of the regularity of the Lagrangian function.

A solution of the Herglotz--Euler--Lagrange equations satisfies $\dot{\xi}= X_L\circ\xi$, with $X_L$ a second-order Lagrangian vector field defined on an appropriate submanifold of $\Tan Q \times \R$. Thus, an immediate consequence of the last proposition is that we can write
\begin{equation}
    \Tan(\FL) \circ X_L \circ \xi = K\circ \xi\,.
\end{equation}
If $S_f$ is the final constraint submanifold, since we have solutions at every point, we can write
\begin{equation}
    \restr{K}{S_f} = \Tan(\FL) \circ \restr{X_L}{S_f} \,.
\end{equation}
If the Lagrangian is regular, then the Legendre transformation is a local diffeomorphism, and the Lagrangian vector field $X_L$ is uniquely determined. Therefore, in the regular case, we obtain
\begin{equation}
    X_L=\Tan(\FL^{-1}) \circ K \,.
\end{equation}

\begin{proposition}\label{H-Dirac}
    Let $\psi\colon I\to \Tan^*Q \times \R$ be a path in the extended cotangent bundle and consider its canonical lift $\dot \psi \colon I\to \Tan(\Tan^*Q \times \R)$. Then $\psi$ is a solution to the Herglotz--Dirac equations~\eqref{Herglotz--Dirac_eq} if, and only if,
    \begin{equation}
        \dot \psi =  K \circ \rho \circ \Tan(\pi_0) \circ \dot \psi \,,
    \end{equation}
    where $\pi_0 \colon \Tan^* Q \times \R \to Q \times \R$ and $\rho\colon \Tan(Q\times \R) \to \Tan Q \times \R$ are the canonical projections.
\end{proposition}

\begin{proof}
    It is enough to explicitly write both sides of the equation in coordinates. If $\psi=(q,p,s)$, then its canonical lift is expressed as $\dot\psi=(q,p,s;\dot q, \dot p,\dot s)$. Therefore,
    \begin{align}
        K\circ \rho \circ \Tan(\pi_0) \circ \dot\psi = \left( q, \parder{L}{v},s\, ; \, \dot q,\parder{L}{q}+\parder{L}{v}\parder{L}{s},L  \right)\,.
    \end{align}
    If we compare both expressions, we obtain
    \begin{equation}
    \begin{dcases}
    \dot s = L\,,
    \\
     p=\parder{L}{v}\,,
    \\
    \dot p = \parder{L}{q} + \parder{L}{v} \parder{L}{s} \,,
    \end{dcases}
    \end{equation}
    which are precisely the Herglotz--Dirac equations~\ref{Herglotz--Dirac_eq} for the Lagrangian $L$.
\end{proof}

Since $K$ is an embedding, there is a correspondence between the solutions of the Herglotz--Euler--Lagrange and the Herglotz--Dirac equations. The map $\xi \mapsto \FL \circ\xi$ sends solutions to solutions, and its inverse is $\psi \mapsto \rho \circ \Tan(\pi_0)\circ \dot\psi$.
Thus, if there exists a solution to the Herglotz--Dirac equations $\psi \colon I \to \Tan^*Q \times \R$, it can be expressed as 
$\psi = \FL \circ \xi$, where $\xi \colon I \to \Tan Q\times \R$ is a solution of the Herglotz--Euler--Lagrange equations, and is obtained as the prolongation of the projected path $\pi_0 \circ \psi$.
By Propositions \ref{contact_K_XL_relation} and \ref{H-Dirac}, we have
\begin{equation}\label{contact_K_relation_1}
    \dot\psi = \Tan (\FL) \circ \dot{\xi} = K\circ \xi = K \circ \rho \circ \Tan (\pi_0) \circ \dot \psi. 
\end{equation}

Now, suppose that we have a Hamiltonian vector field $X_H$, defined on the image of the Legendre map $\FL(\Tan Q\times \R) \subseteq \Tan^*Q\times \R$, 
for the Hamiltonian formalism associated with the Lagrangian system $(\Tan Q \times \R, L)$. Then, the solutions to Hamilton's equations are its integral curves 
$\dot\psi=X_H \circ \psi$, and we can write
\begin{equation} \label{contact_K_relation_2}
    K\circ \xi=\dot\psi = X_H \circ \psi = X_H \circ \FL \circ \xi\,.
\end{equation}
In the final constraint submanifold $S_f$, this gives the relation
\begin{equation}
    \restr{K}{S_f} = X_H\circ \restr{\FL}{S_f} \,.
\end{equation}
And, if the Lagrangian is regular, we have
$$X_H = K\circ \FL^{-1}\,.$$

Lastly, assume that we have two $\FL$-related solutions $\xi$ and $\psi$ of the Herglotz--Euler--Lagrange and the Herglotz--Dirac equations, respectively, i.e.\ $\psi = \FL \circ {\xi}$
and $\xi = \rho\circ\Tan(\pi_0) \circ \dot\psi$. Then, for a given function $f\in\Cinfty(\Tan^* Q\times \R)$,
\begin{equation}
    \frac{\d}{\d t}(f\circ \psi) = \langle \d f \circ \psi, \dot\psi \rangle = \langle \d f \circ (\FL \circ  \xi), K\circ {\xi} \rangle = (K\cdot f ) \circ \xi \,,
\end{equation}
where we used Equations \eqref{contact_K_derivation} and \eqref{contact_K_relation_2}.

\begin{corollary}\label{K_contact_constraints}
Suppose that \( f \in \Cinfty(\Tan^*Q \times \mathbb{R}) \) is a Hamiltonian constant of motion, such as a Hamiltonian constraint. Then $K\cdot f$ is a Lagrangian constraint.
\end{corollary}

\subsection{Relating the constraint algorithms}

Let $L\in \Cinfty (\T Q \times \R)$ be an almost-regular Lagrangian function, so we can apply the results presented in Section~\ref{sec:almost_regular}.
Recall that $P_0 \coloneqq \FL(\Tan Q\times \R)$ is a closed submanifold, (locally) described by the vanishing of $m$ functions $\{\phi_{\mu}\}_{\mu=1,\dots,m}$, with linearly independent differentials $\d \phi_{\mu}$ at every point of $P_0$. Also, the vertical vector fields $\Gamma_{\mu}\in\X(\Tan Q\times \R)$, locally given by
$$\Gamma_{\mu} = \gamma_{\mu}\parder{}{v}=\FL^*\left( \parder{\phi_{\mu}}{p} \right)\parder{}{v}\,,$$
provide a frame for the vector subbundle $\Ker \Tan(\FL)$. 

Since the Lagrangian is almost-regular, there exists a unique function $H_0 \in \Cinfty(P_0)$ such that $E_L = H_0\circ\FL$ on $P_0$. 
The function $H_0$ can be extended to a function $H\in\Cinfty(\Tan^*Q\times \R)$ defined on the whole extended phase space $\Tan^*Q \times \R$. This function $H$ satisfies $\FL^*(H)=E_L$ and defines a (unique) contact Hamiltonian vector field $X_H$.

Our aim is to write the evolution operator $K$ in terms of the contact Hamiltonian vector fields associated with a choice of Hamiltonian function and a set of primary Hamiltonian constraints. 
\begin{proposition} \label{K_XH_expression}
	Let $H\in \Cinfty (\Tan^*Q)$ be a fixed Hamiltonian function such that $\FL^*(H)=E_L$, and $\{\phi_{\mu}\}_{\mu}$ be a set of primary constraint Hamiltonian functions. Then, there locally exist $m$ uniquely determined functions $\lambda^{\mu}$ such that
	\begin{equation}
		K = X_H \circ \FL + \sum_{\mu= 1}^m \lambda^{\mu} (X_{\phi_\mu} \circ \FL)\,.
	\end{equation}
    These functions $\lambda^{\mu}$ are the same as those appearing in Proposition~\ref{prop_resolution_identity}.
\end{proposition}
\begin{proof}
    Using $\FL^* (H) = E_L$, we can write \eqref{K_formula_def} as
\begin{align}
        ^t \T \FL\left(B\circ K -\d H \circ \FL + (H \circ\FL) (\eta_Q\circ \FL)- \parder{L}{s}(\eta_Q\circ \FL)\right) = 0\,.
\end{align}
Hence, we have a 1-form along the Legendre map that belongs to $\Ker ^t\Tan(\FL)$. Recall that $\{\d \phi_{\mu}\circ \FL\}_\mu$ is a frame for $\Ker ^t\Tan(\FL)$, and so
$$B\circ K -\d H \circ \FL + (H \circ\FL) (\eta_Q\circ \FL)- \parder{L}{s}(\eta_Q\circ \FL)=\sum_{\mu} \alpha^{\mu} \d\phi_{\mu}\circ\FL\,,$$
where the $\alpha^{\mu}$ are some uniquely  determined functions. 

Applying $B^{-1}$ to both sides and rearranging terms, we obtain
$$K= B^{-1}(\d H) \circ \FL + \left( \parder{L}{s}  -H \circ\FL\right)R\circ \FL + \sum_{\mu}\alpha^{\mu} B^{-1}(\d\phi_{\mu})\circ\FL\,.$$
Using Equation~\eqref{resolution_identity} together with the condition 
$\Id_{\Tan Q \times \R} = \rho\circ \Tan\pi_0 \circ K$, one can check, in local coordinates, that the functions $\alpha^{\mu}$ coincide with the functions $\lambda^{\mu}$ from Proposition~\ref{prop_resolution_identity}.
Thus, we have
\begin{equation}\label{K_formula_intermidiate}
    K= B^{-1}(\d H) \circ \FL + \left( \parder{L}{s}  -H \circ\FL\right)R\circ \FL + \sum_{\mu}\lambda^{\mu} B^{-1}(\d\phi_{\mu})\circ\FL\,.
\end{equation}
Note that
\begin{equation}
    K\cdot s =L = \FL^*\left(p_i\parder{H}{p_i}+\parder{H}{s}\right) - H\circ \FL +\parder{L}{s} +\sum_{\mu} \lambda^{\mu} \FL^*\left(p_i\parder{\phi_\mu}{p_i}+\parder{\phi_{\mu}}{s}\right)\,,
\end{equation}
where we have used the local expression of the vector field defined by $\displaystyle B^{-1}(\d f)$, for any function $f\in \Cinfty (\Tan^*Q \times \R)$.

Since $E_L = H\circ \FL$ and $\FL^*(p_i)=\tparder{L}{v^i}$, from the last expression we obtain
\begin{equation}
    E_L+L=\Delta(L) = \FL^*\left(\parder{H}{p_i}\right)\fracp{L}{v^i}+\FL^*\left(\parder{H}{s}\right) +\parder{L}{s} +\sum_{\mu} \lambda^{\mu} \left(\FL^*\left(\parder{\phi_\mu}{p_i}\right)\fracp{L}{v^i}+\FL^*\left(\parder{\phi_{\mu}}{s}\right)\right)\,,
\end{equation}
and, combining this with~\eqref{resolution_liouville}, which yields
\begin{equation}
    \Delta(L) = \Gamma_H(L) +\sum_{\mu}\lambda^{\mu} \Gamma_{\mu}(L) = \FL^*\left(\parder{H}{p_i}\right)\parder{L}{v^i}+\sum_{\mu}\lambda^{\mu}\FL^*\left(\parder{\phi_{\mu}}{p_i}\right)\parder{L}{v^i}\,,
\end{equation}
we immediately obtain the expression
\begin{equation}\label{L_partial_s_formula}
    -\parder{L}{s} = \FL^*\left(\parder{H}{s}\right) + \sum_{\mu}\lambda^{\mu}\FL^*\left( \parder{\phi_{\mu}}{s} \right)\,.
\end{equation}

Now, the result follows directly from Equations~\eqref{K_formula_intermidiate} and~\eqref{L_partial_s_formula}, and the fact that for every primary Hamiltonian constraint we have $\phi_{\mu}\circ \FL = 0$, by definition.
\end{proof}

With this result, for every $f \in \Cinfty(\Tan^*Q\times\R)$, we can write $K\cdot f \in \Cinfty(\Tan Q\times \R)$ as
\begin{equation}
    K\cdot f = \FL ^* (X_H\cdot f) + \sum _{\mu = 1} ^m \lambda ^{\mu} \FL^*(X_{\phi_\mu} \cdot  f)\,.
\end{equation}
In particular, from the local expression of $K$ we have
\begin{align} 
      v &= K\cdot q = \FL ^* (X_H\cdot q) + \sum _{\mu = 1} ^m \lambda ^{\mu} \FL^*(X_{\phi_\mu}\cdot  q)\,, 
  \\
    \parder{L}{q} +\parder{L}{s}\parder{L}{v} &= K\cdot p = \FL ^* (X_H\cdot p) + \sum _{\mu = 1} ^m \lambda ^{\mu} \FL^*(X_{\phi_\mu}\cdot p)\,,
    \\
    L &= K\cdot s= \FL ^* (X_H\cdot s) + \sum _{\mu = 1} ^m \lambda ^{\mu} \FL^*(X_{\phi_\mu}\cdot s)\,.
\end{align}

\begin{proposition}
     Let $f$ be a function on the phase space. The function $K\cdot f$ is $\FL$-projectable if, and only if, $X_{\phi_{\mu}}\cdot f=0$ for all primary Hamiltonian constraints $\phi_{\mu}$. 
\end{proposition}
\begin{proof}
    We have
    $$\Gamma_{\mu}\cdot(K\cdot f) = \sum_{\nu=1}^m \FL^*(X_{ \phi_{\nu}}\cdot f)(\Gamma_{\mu} \cdot {\lambda}^{\nu}) \,,$$
    and, applying \eqref{Gamma_vf_condition}, we get
    $$\Gamma_{\mu}\cdot(K\cdot f) =\sum_{\mu=1}^m \FL^*(X_{\phi_{\mu}}\cdot f)\,.$$
    Thus, $\Gamma_\mu\cdot(K\cdot f)=0$ for every $\mu=1,\dots,m$ if, and only if, $(X_{\phi_{\mu}}\cdot f)=0$, as we wanted to see.
\end{proof}

\section{Examples}\label{sec:examples}

This last section is devoted to the study of a couple interesting examples: the simple pendulum and Cawley's Lagrangian, both with an extra dissipation term.

We use the theory of linearly singular systems (see \cite{GP_92} for details), and the alternative contact Hamiltonian equations without Reeb vector fields~\eqref{cH_eq_1} to develop the constraint algorithm.

\subsection{The simple pendulum with damping}

Consider a damped pendulum of length $\ell$ and mass $m$. Its position in the plane is given using polar coordinates $(r,\theta)$, where $\theta = 0$ is the rest position. The Lagrangian phase space has coordinates $(r,\theta,\mu; v_r,v_\theta,v_\mu;s)$ and the Lagrangian function of the system $L\colon\T Q\times\R\to\R$ is given by
$$ L = \frac{1}{2}m(v_r^2 + r^2v_\theta^2) - mg(\ell - r\cos\theta) + \mu(r - \ell) - \gamma s\,, $$
where $\gamma\in\R$ is a damping coefficient. 
The Legendre map associated with this Lagrangian function is
$$ \F L(r,\theta,\mu; v_r,v_\theta,v_\mu;s) = \left( r,\theta,\mu;\, p_r = mv_r,\, p_\theta = mr^2v_\theta,\, p_\mu = 0, s \right)\,, $$
giving the constraint
$$ \phi_0 = p_\mu\,. $$
The Lagrangian energy is
\begin{equation}
    E_L=\frac{1}{2}m(v_r^2 + r^2v_\theta^2) + mg(\ell - r\cos\theta) - \mu(r - \ell) + \gamma s \, ;
\end{equation}
so we can take as a possible Hamiltonian function
\begin{equation}
    H= \frac{1}{2m}\left(p_r^2 + \frac{p_\theta^2}{r^2}\right) + mg(\ell - r\cos\theta) - \mu(r - \ell) + \gamma s
\end{equation}
We can compute the corresponding Hamiltonian vector fields, with respect to the 1-form defined in $P_0\coloneq\FL(\Tan Q \times \R)$ by taking $\eta_0 \coloneqq j^*(\eta_Q)$, where $j\colon P_0\hookrightarrow \Tan^*Q \times \R$. In coordinates, we have
$$
\eta_0 = \d s - p_r \d r -p_{\theta}\d \theta\,,
$$
and, using Equations~\eqref{cH_eq_1}, we obtain
\begin{multline}
    X_H = \frac{p_r}{m}\parder{}{r}+\frac{p_{\theta}}{mr^2}\parder{}{\theta}+f_1\parder{}{\mu} +\left(\frac{p_{\theta}^2}{mr^3}+mg\cos\theta + \mu -p_r\gamma\right)\parder{}{p_r}  \\+(-mgr\sin\theta -p_{\theta}\gamma)\parder{}{p_\theta}+f_2\parder{}{p_{\mu}}+ \left(\frac{p_r^2}{m}+ \frac{p_\theta^2}{mr^2}-H \right)\parder{}{s}\,,
\end{multline}
where $f_1,f_2$ are any arbitrary functions. We also obtain, as a necessary condition, the constraint
$$ \phi_1 \coloneqq r-\ell\,. $$
Now, when we perform the constraint algorithm, the tangency condition $0 = X_H \cdot \phi_0 = f_2$ determines the coefficient $f_2$. We also obtain the new constraint 
$$ \phi_2\coloneqq X_H \cdot \phi_1 = \frac{p_r}{m}\,. $$
If we keep performing the algorithm, this last constraint yields
$$\phi_3\coloneqq X_H\cdot\phi_2 = \frac{1}{m}\left(\frac{p_{\theta}^2}{mr^3} + mg\cos\theta + \mu -p_r\gamma\right)\,,$$
and, finally, the tangency condition $X_H\cdot \phi_3 = 0$ determines the coefficient $f_1$. The constraint algorithm ends here.

The evolution operator associated with this Lagrangian function has local expression
\begin{multline}
    K = v_r\parder{}{r}\Big|_{\FL} + v_\theta\parder{}{\theta}\Big|_{\FL} + v_\mu\parder{}{\mu}\Big|_{\FL} + \left( mrv_\theta^2 + mg\cos\theta + \mu - \gamma mv_r \right)\parder{}{p_r}\Big|_{\FL} \\+ \left( -mgr\sin\theta - \gamma mr^2v_\theta \right)\parder{}{p_\theta}\Big|_{\FL} + (r-\ell)\parder{}{p_\mu} + L\parder{}{s}\Big|_{\FL} \,.
\end{multline}
We can apply the evolution operator to the Hamiltonian constraints to obtain Lagrangian ones. Namely, we obtain the following constraints,
\begin{equation}
    \begin{dcases}
        \chi_1 \coloneqq K\cdot \phi_0 = r-\ell \,,
        \\
        \chi_2 \coloneqq K\cdot \phi_1 =  v_r\,,
        \\
        \chi_3 \coloneqq K\cdot \phi_2 = rv_\theta^2 + g\cos\theta + \mu/m -\gamma v_r\,,
        \\
        \chi_4\coloneqq K\cdot \phi_3 = -2v_{\theta}(g\sin\theta + \gamma rv_{\theta})-3v_rv_{\theta}^2 +\frac{v_\mu}{m} -\frac{\gamma}{m}(mrv_{\theta}^2+mg\cos\theta+\mu -\gamma m v_r)\,.
    \end{dcases}
\end{equation}
If one performs the constraint algorithm directly for the Lagrangian formalism, one can check that these same constraints are obtained.

We can also find the functions $\lambda$, from~\eqref{resolution_liouville}. We have
\begin{equation}
    \Gamma_ H = v_r \parder{}{v_r} + v_{\theta} \parder{}{v_{\theta}}\,, \qquad \Gamma_0 = \parder{}{v_\mu},
\end{equation}
and so,
$$\Delta = \Gamma_H + v_{\mu}\Gamma_0\,,$$
which implies that $\lambda^0 = v_{\mu}$ in this case.

If we find $X_H, X_{\phi_0}$, the contact Hamiltonian vector fields, now with respect to the canonical contact structure $\eta_Q \in \Omega^1(\Tan ^* Q\times \R)$, we obtain
\begin{multline}
    X_H = \frac{p_r}{m}\parder{}{r}+\frac{p_{\theta}}{mr^2}\parder{}{\theta}+\left(\frac{p_{\theta}^2}{mr^3}+mg\cos\theta + \mu -p_r\gamma\right)\parder{}{p_r}  \\+(-mgr\sin\theta -p_{\theta}\gamma)\parder{}{p_\theta}+(r-\ell +p_{\mu}\gamma)\parder{}{p_{\mu}}+ \left(\frac{p_r^2}{m}+ \frac{p_\theta^2}{mr^2}-H \right)\parder{}{s}\,,
\end{multline}
and 
\begin{equation}
     X_{\phi_0} = \parder{}{\mu}\,.
\end{equation}
One can easily check that 
$$K = X_H  \circ \FL + v_{\mu} X_{\phi_0}\circ \FL\,.$$

\subsection{Cawley's Lagrangian with dissipation}

Cawley's Lagrangian is an academic model introduced by R.\,Cawley to study some features of singular Lagrangians in Dirac's theory of constraint systems \cite{Caw_79}. 
In this example we introduce a velocity-dependent dissipation term. Consider the manifold $\Tan\R^3\times\R$ with canonical coordinates $(x,y,z;v_x,v_y,v_z;s)$ and the Lagrangian function
$$ 
L = v_xv_z + \frac{1}{2}yz^2 - \gamma s v_y
\,,
$$
where $\gamma$ is a non-zero damping coefficient. 
The Legendre map is
$$ 
\FL\colon(x,y,z;v_x,v_y,v_z;s) \longmapsto 
(x,y,z;\,p_x=v_z,\,p_y = -\gamma s,\,p_z = v_x;\,s)  \,.
$$
The primary Hamiltonian constraint is 
$\phi_0 = p_y+ \gamma s$,
and $L$ is an almost-regular Lagrangian.

The Lagrangian energy is
$$ E_L = v_xv_z - \frac{1}{2}yz^2\,, $$
therefore, we can take as a Hamiltonian function
$$
H = p_x p_z - \frac{1}{2}yz^2\,.
$$

Now, we can compute the corresponding Hamiltonian vector fields with respect to the 1-form defined on $P_0\coloneq\FL(\Tan Q \times \R)$ by taking $\eta_0 \coloneqq j^*(\eta_Q)$, where $j\colon P_0\hookrightarrow \Tan^*Q \times \R$. In coordinates, we have
$$
\eta_0 = \d s - p_x \d x + \gamma s\d y - p_z\d z \,.
$$
One can check that this 1-form \textit{does not define any Reeb vector field}, and so the usual contact Hamiltonian equations cannot be used to find the corresponding Hamiltonian vector fields in this case.
However, if we apply the alternative Hamiltonian equations~\eqref{cH_eq_1}, we obtain
$$
X_H =  
p_z \parder{}{x} +b\parder{}{y}+ p_x \parder{}{z} -\gamma b p_x \parder{}{p_x}
+ c \parder{}{p_y} + (yz- \gamma b p_z) \parder{}{p_z} 
+ \left(p_xp_z + \frac12 yz^2 -\gamma b s \right) \parder{}{s}
\,,
$$
where $b$ and $c$ are arbitrary functions. We also obtain as a necessary condition the constraint 
$$ \phi_1 \coloneqq \frac12 z ^2 + \gamma\left(p_x p_z +\frac{1}{2}yz^2\right)\,. $$
Now, demanding the tangency to the first constraint submanifold
$$X_H\cdot \phi_0 = c +\gamma\left(p_xp_z +\frac12 yz^2-\gamma bs \right)  = 0$$
we obtain an expression for $c$, in terms of the function $b$.
Again, the same can be done for the constraint $\phi_1$, and the condition $X_H \cdot \phi_1=0$ will determine the function $b$. The constraint algorithm ends here.

The evolution operator is given by
\begin{align}
    K &= 
v_x\parder{}{x}\Big|_{\FL}  + v_y\parder{}{y}\Big|_{\FL}  + v_z\parder{}{z}\Big|_{\FL} 
-\gamma v_yv_z\parder{}{p_x}\Big|_{\FL} \\
&\quad+ \left( \frac{1}{2}z^2 + \gamma^2 s v_y \right) \parder{}{p_y}\Big|_{\FL} 
+ (yz -\gamma v_xv_y)\parder{}{p_z} \Big|_{\FL} 
+ L\parder{}{s}\Big|_{\FL} 
\,. 
\end{align}
We can apply the evolution operator to the Hamiltonian constraints to obtain Lagrangian ones. Namely, we obtain the following constraints,
\begin{equation}
    \begin{dcases}
        \chi_1 \coloneqq K\cdot \phi_0 = \frac{1}{2}z^2 + \gamma \left(v_x v_z +\frac{1}{2}yz^2 \right)\,,
        \\
        \chi_2 \coloneqq K\cdot \phi_1 = zv_z(1+2 \gamma y) + \gamma v_y \left( \frac{1}{2}z^2   -2\gamma v_x v_z \right)\,.
    \end{dcases}
\end{equation}

\section{Conclusions and outlook}
\label{sec:conclusions}

In this paper, we have introduced and studied the evolution operator $K$, for singular contact Lagrangian systems. 
To this end, we first reviewed the characterization of the evolution operator~$K$ in ordinary mechanics. We also reviewed the basic facts of contact geometry and contact Hamiltonian systems, for which we obtained alternative formulations of the dynamical equations that do not rely on the Reeb vector field.

We presented a definition of the evolution operator~$K$ within the framework of contact mechanics. This operator exhibits properties analogous to those in ordinary symplectic mechanics, and provides a fundamental link between the Lagrangian and Hamiltonian descriptions in the singular contact setting. To illustrate the theoretical developments, we examined two explicit examples: the simple pendulum and the Cawley Lagrangian, both with extra damping terms.

There are several aspects of singular contact dynamics that deserve further investigation. 
For singular Lagrangian systems the operator~$K$ provides a tight connection between the Lagrangian and Hamiltonian formalisms. 
We plan to explore in detail these connections for contact Lagrangian systems in the future.
This also includes the development of constraint algorithms for precontact Lagrangian and Hamiltonian systems. In order to provide a good account of precontact systems, we also aim to study precontact structures in detail.

\addcontentsline{toc}{section}{Acknowledgements}
\section*{Acknowledgements}

We acknowledge the financial support from the Spanish Ministry of Science and Innovation, grants  ID2021-125515NB-C21 (MCIN/AEI/10.13039/501100011033/FEDER,UE), and RED2022-134301-T of AEI, 
and the Ministry of Research and Universities of
the Catalan Government, project 2021 SGR 00603 \textsl{Geometry of Manifolds and Applications, GEOMVAP}.

AMM acknowledges financial support from a predoctoral contract funded by Universitat Rovira i Virgili under grant 2025PMF-PIPF-14.

\bibliographystyle{abbrv}
{\small
\bibliography{references.bib}

@book{AM_78,
    author = {Ralph Abraham and Jerrold E. Marsden},
    title = {{Foundations of mechanics}},
    publisher = {Benjamin/Cummings Pub. Co.},
    series={AMS Chelsea publishing},
    volume = {364},
    year = {1978},
    edition = {2nd},
    address = {New York},
    note = {\href{https://doi.org/10.1090/chel/364}{10.1090/chel/364}}
}

@book{Bou_Var,
    author = {N. Bourbaki},
    title = {Variétés différentielles et analytiques (fascicule de résultats)},
    publisher = {Springer},
    year = {2007},
    series = {\'El\'ements de mathématique},
    note = {\href{https://doi.org/10.1007/978-3-540-34397-4}{10.1007/978-3-540-34397-4}}
}

@article{Bra_17,
    author = {Alessandro Bravetti},
    title = {{Contact Hamiltonian dynamics: The concept and its use}},
    journal = {Entropy},
    volume = {{\bf 10}},
    number = {19},
    pages = {535},
    year = {2017},
    note = {\href{https://doi.org/10.3390/e19100535}{10.3390/e19100535}}
}

@article{BCT_17,
    author = {Alessandro Bravetti and Hans Cruz and Diego Tapias},
    title = {{Contact Hamiltonian mechanics}},
    journal = {Ann. Phys.},
    volume = {{\bf 376}},
    pages = {17--39},
    year = {2017},
    note = {\href{https://doi.org/10.1016/j.aop.2016.11.003}{10.1016/j.aop.2016.11.003}}
}

@article{BGG_17,
    author = {A. J. Bruce and K. Grabowska and J. Grabowski},
    title = {{Remarks on Contact and Jacobi Geometry}},
    journal = {Symmetry Integr. Geom.: Methods Appl. (SIGMA)},
    volume = {{\bf 13}},
    year = {2017},
    pages = {059},
    note = {\href{https://doi.org/10.3842/SIGMA.2017.059}{10.3842/SIGMA.2017.059}}
}

@article{BGPR_87,
    author = {Carles Batlle and Joaquim Gomis and Josep M. Pons and Narciso Román-Roy},
    title = {{Lagrangian and Hamiltonian constraints}},
    journal = {Lett. Math. Phys.},
    volume = {{\bf 13}},
    pages = {17--23},
    year = {1987},
    note = {\href{https://doi.org/10.1007/BF00570763}{10.1007/BF00570763}}
}

@article{BGPR_86,
    author = {Carles Batlle and Joaquim Gomis and Josep M. Pons and Narciso Román-Roy},
    title = {{Equivalence between the Lagrangian and Hamiltonian formalism for constrained systems}},
    journal = {J. Math. Phys.},
    volume = {{\bf 27}},
    number = {12},
    pages = {2953--2962},
    year = {1986},
    note = {\href{https://doi.org/10.1063/1.527274}{10.1063/1.527274}}
}

@article{BGPR_1988,
note = 
{\href{https://doi.org/10.1088/0305-4470/21/12/013}{10.1088/0305-4470/21/12/013}},
year = {1988},
publisher = {},
volume = {21},
number = {12},
pages = {2693},
author = {C. Batlle and J. Gomis and J. M. Pons and N. Román-Roy},
title = {{Lagrangian and Hamiltonian constraints for second-order singular Lagrangians}},
journal = {J. Phys. A: Math. Gen.}
}

@book{BH_16,
    author = {Augustin Banyaga and Djideme F. Houenou},
    title = {{A brief introduction to symplectic and contact manifolds}},
    address = {Singapore},
    series = {Nankai Tracts in Mathematics},
    publisher = {World Scientific Publishing Co. Pte. Ltd.},
    volume = {15},
    year = {2016},
    note = {\href{https://doi.org/10.1142/9667}{10.1142/9667}}
}

@article{BLMP_20,
    author = {Alessandro Bravetti and Manuel de León and Juan Carlos Marrero and Edith Padrón},
    title = {{Invariant measures for contact Hamiltonian systems: symplectic sandwiches with contact bread}},
    journal = {J. Phys. A: Math. Theor.},
    volume = {{\bf 53}},
    pages = {455205},
    year = {2020},
    note = {\href{https://doi.org/10.1088/1751-8121/abbaaa}{10.1088/1751-8121/abbaaa}}
}

@article{Car_96,
    author = {José F. Cariñena},
    title = {{Sections along maps in geometry and physics}},
    journal = {Rend. Sem. Mat. Univ. Pol. Torino},
    volume = {{\bf 54}},
    number = {3},
    pages = {245--256},
    year = {1996},
    note = {\href{https://arxiv.org/abs/dg-ga/9703017}{10.48550/arXiv.dg-ga/9703017}}
}

@article{Caw_79,
    author = {Robert Cawley},
    title = {{Determination of the Hamiltonian in the presence of constraints}},
    journal = {Phys. Rev. Lett.},
    volume = {{\bf 42}},
    number = {7},
    pages = {413--416},
    year = {1979},
    note = {\href{https://doi.org/10.1103/PhysRevLett.42.413}{10.1103/PhysRevLett.42.413}}
}

@article{CCM_18,
    author = {Florio M. Ciaglia and Hans Cruz and Giuseppe Marmo},
    title = {{Contact manifolds and dissipation, classical and quantum}},
    journal = {Ann. Phys.},
    volume = {{\bf 398}},
    pages = {159--179},
    year = {2018},
    note = {\href{https://doi.org/10.1016/j.aop.2018.09.012}{10.1016/j.aop.2018.09.012}}
}

@article{CG_19,
    author = {José F. Cariñena and P. Guha},
    title = {{Nonstandard Hamiltonian structures of the Liénard equation and contact geometry}},
    journal = {Int. J. Geom. Methods Mod. Phys.},
    volume = {{\bf 16}},
    number = {supp01},
    pages = {1940001},
    year = {2019},
    note = {\href{https://doi.org/10.1142/S0219887819400012}{10.1142/S0219887819400012}}
}

@article{CL_87,
    author = {José F. Cariñena and Carlos López},
    title = {{The time-evolution operator for singular Lagrangians}},
    journal = {Lett. Math. Phys.},
    year = {1987},
    volume = {{\bf 14}},
    pages = {203--210},
    note = {\href{https://doi.org/10.1007/BF00416849}{10.1007/BF00416849}}
}

@article{CL_92,
    title = {The time-evolution operator for higher-order singular Lagrangians},
    author = {José F. Cariñena and Carlos López},
    journal = {Int. J. Mod. Phys. A},
    volume = {{\bf 7}},
    number = {11},
    pages = {2447--2468},
    year = {1992},
    note = {\href{https://doi.org/10.1142/S0217751X92001083}{10.1142/S0217751X92001083}}
}

@article{CLM_89,
    author = {José F. Cariñena and Carlos López and Eduardo Martínez},
    title = {{A new approach to the converse of Noether's theorem}},
    journal = {J. Phys. A: Math. Gen.},
    year = {1989},
    volume = {{\bf 22}},
    number = {22},
    pages = {4777},
    note = {\href{https://doi.org/10.1088/0305-4470/22/22/009}{10.1088/0305-4470/22/22/009}}
}

@article{Dir_50,
    author = {P. A. M. Dirac},
    title = {{Generalized Hamiltonian Dynamics}},
    journal = {Can. J. Math.},
    volume = {{\bf 2}},
    year = {1950},
    pages = {129--148},
    note = {\href{https://doi.org/10.4153/CJM-1950-012-1}{10.4153/CJM-1950-012-1}}
}

@article{EMMR_03,
    author = {Arturo Echeverría-Enríquez and Jesús Marín-Solano and Miguel Carlos Muñoz-Lecanda and Narciso Román-Roy},
    title = {{On the construction of $K$-operators in field theories as sections along Legendre maps}},
    journal = {Acta Appl. Math.},
    volume = {{\bf 77}},
    pages = {1--40},
    year = {2003},
    note = {\href{https://doi.org/10.1023/A:1023671402908}{10.1023/A:1023671402908}}
}

@book{Gei_08,
    title = {{An introduction to contact topology}},
    publisher = {Cambridge University Press},
    year = {2008},
    author = {Hansjörg Geiges},
    series = {Cambridge Studies in Advanced Mathematics},
    address = {New York, NY},
    volume = {109},
    note = {\href{https://doi.org/10.1017/CBO9780511611438}{10.1017/CBO9780511611438}}
}

@article{Gra_00,
    author = {Xavier Gràcia},
    title = {{Fibre derivatives: Some applications to singular Lagrangians}},
    journal = {Rep. Math. Phys.},
    volume = {{\bf 45}},
    number = {1},
    pages = {67--84},
    year = {2000},
    note = {\href{https://doi.org/10.1016/S0034-4877(00)88872-2}{10.1016/S0034-4877(00)88872-2}}
}

@article{GN_79,
  author  = {M.J. Gotay and J.M. Nester},
  title   = {{Presymplectic Lagrangian systems. I : the constraint algorithm and the equivalence theorem}},
  journal = {Ann. Inst. Henri Poincar\'e A, Phys. Th\'eor.},
  year    = {1979},
  volume  = {{\bf 30}},
  number  = {2},
  pages   = {129--142},
  note    = {\url{https://www.numdam.org/item/?id=AIHPA_1979__30_2_129_0}}
}

@article{Ber_51,
  title = {{Constraints in Covariant Field Theories}},
  author = {Anderson, James L. and Bergmann, Peter G.},
  journal = {Phys. Rev.},
  volume = {{\bf 83}},
  issue = {5},
  pages = {1018--1025},
  numpages = {0},
  year = {1951},
  month = {Sep},
  publisher = {American Physical Society},
  doi = {10.1103/PhysRev.83.1018},
  url = {https://link.aps.org/doi/10.1103/PhysRev.83.1018},
note = {\href{ https://doi.org/10.1103/PhysRev.83.1018}{10.1103/PhysRev.83.1018}}
}

@article{GG_22,
    author = {Katarzyna Grabowska and Janusz Grabowski},
    title = {{A geometric approach to contact Hamiltonians and contact Hamilton--Jacobi theory}},
    journal = {J. Phys. A: Math. Theor.},
    volume = {{\bf 55}},
    number = {43},
    pages = {435204},
    year = {2022},
    note = {\href{https://doi.org/10.1088/1751-8121/ac9adb}{10.1088/1751-8121/ac9adb}}
}

@article{GG_23,
    author = {Katarzyna Grabowska and Janusz Grabowski},
    title = {{Reductions: precontact versus presymplectic}},
    journal = {Ann. Mat. Pura Appl.},
    volume = {{\bf 202}},
    pages = {2803--2839},
    year = {2023},
    note = {\href{https://doi.org/10.1007/s10231-023-01341-y}{10.1007/s10231-023-01341-y}}
}

@article{GGMRR_20a,
    author = {Jordi Gaset and Xavier Gràcia and Miguel C. Muñoz-Lecanda and Xavier Rivas and Narciso Román-Roy},
    title = {{New contributions to the Hamiltonian and Lagrangian contact formalisms for dissipative mechanical systems and their symmetries}},
    journal = {Int. J. Geom. Methods Mod. Phys.},
    volume = {{\bf 17}},
    number = {6},
    pages = {2050090},
    year = {2020},
    note = {\href{https://doi.org/10.1142/S0219887820500905}{10.1142/S0219887820500905}}
}

@inproceedings{CGMMMR_2009,
author = {J.F. Cariñena and X. Gràcia and E. Martínez and G. Marmo and M.-C. Muñoz-Lecanda and N. Román-Roy},
title = {{Hamilton--Jacobi theory and the evolution operator}}, 
booktitle = {Mathematical physics and field theory. Julio Abad, in memoriam},
year = {2009},
pages = {177--186},
editor = {M. Asorey and J.V. García Esteve and M.F. Rañada and J. Sesma},
publisher = {Prensas Universitarias de Zaragoza},
ISBN = {978-84-92774-04-3},
note={\href{https://arxiv.org/abs/0907.1039}{0907.1039}}
}

@article{PV_00,
  title={{On the geometry of singular Lagrangians}},
  author={Fabrizio Pugliese and Alexandre M. Vinogradov},
  journal={J. Geom. Phys.},
  year={2000},
  volume={{\bf 35}},
  pages={35-55},
  url={https://api.semanticscholar.org/CorpusID:120620170},
 note={\href{https://doi.org/10.1016/S0393-0440(99)00076-5}{10.1016/S0393-0440(99)00076-5}}
}

@article{GraPon_1988,
    author = {Xavier Gràcia and Josep M. Pons},
    title = {{Gauge generators, Dirac's conjecture, and degrees of freedom for constrained systems}},
    journal = {Ann. Phys.},
    volume = {{\bf 187}},
    pages = {355--368},
    year = {1988},
    note ={\href{https://doi.org/10.1016/0003-4916(88)90153-4}{10.1016/0003-4916(88)90153-4}}
}

@article{GP_89,
    author = {Xavier Gràcia and Josep M. Pons},
    title = {{On an evolution operator connecting Lagrangian and Hamiltonian formalisms}},
    journal = {Lett. Math. Phys.},
    volume = {{\bf 17}},
    pages = {175--180},
    year = {1989},
    note = {\href{https://doi.org/10.1007/BF00401582}{10.1007/BF00401582}}
}

@article{GP_92,
    author = {Xavier Gràcia and Josep M. Pons},
    title = {{A generalized geometric framework for constrained systems}},
    journal = {Diff. Geom. Appl.},
    volume = {{\bf 2}},
    number = {3},
    pages = {223--247},
    year = {1992},
    note = {\href{https://doi.org/10.1016/0926-2245(92)90012-C}{10.1016/0926-2245(92)90012-C}}
}

@article{GP_92a,
    author = {Xavier Gràcia and Josep M. Pons},
    title = {{A Hamiltonian approach to Lagrangian Noether transformations}},
    journal = {J. Phys. A: Math. Gen.},
    volume = {{\bf 25}},
    number = {23},
    pages = {6357},
    year = {1992},
    note = {\href{https://doi.org/10.1088/0305-4470/25/23/029}{10.1088/0305-4470/25/23/029}}
}

@article{GraPon_1994,
     author = {Xavier Gr\`acia and Josep M. Pons},
     title = {Noether transformations with vanishing conserved quantity},
     journal = {Ann. Inst. H. Poicaré A},
     pages = {315--327},
     publisher = {Gauthier-Villars},
     volume = {{\bf 61}},
     number = {3},
     year = {1994},
     mrnumber = {1311070},
     zbl = {0811.70013},
     language = {en},
     note = {\url{https://www.numdam.org/item/AIHPA_1994__61_3_315_0/}}
}

@article{GraPon_1995,
author = {Xavier Gràcia and Josep M. Pons},
title = {Gauge transformations for higher-order {L}agrangians},
journal = {J. Phys. A: Math. Gen.},
volume = {{\bf 28}},
number = {24},
year = {1995},
pages = {7181--7196},
note = {\href{https://doi.org/10.1088/0305-4470/28/24/016}{10.1088/0305-4470/28/24/016}}
}

@article{GraPon_2000,
    author = {Xavier Gràcia and Josep M. Pons},
    title = {{Canonical Noether symmetries and commutativity properties for gauge systems}},
    journal = {J. Math. Phys.},
    volume = {{\bf 41}},
    pages = {7333--7351},
    year = {2000},
    note = {\href{https://doi.org/10.1063/1.1289825}{10.1063/1.1289825}}
}

@phdthesis{Lai_22,
      title={{Contact Hamiltonian systems}},
        school = {ICMAT},
      author={Manuel Lainz},
      year={2022},
      note={\href{https://www.icmat.es/Thesis/2022/Tesis_Manuel_Lainz.pdf}{Link}}, 
}

@article{GraPon_2001,
    author = {Xavier Gràcia and Josep M. Pons},
    title = {{Singular Lagrangians: some geometric structures along the Legendre map}},
    journal = {J. Phys. A: Math. Gen.},
    year = {2001},
    volume = {{\bf 34}},
    number = {14},
    pages = {3047--3070},
    note = {\href{https://doi.org/10.1088/0305-4470/34/14/311}{10.1088/0305-4470/34/14/311}}
}

@article{GPR_91,
    author = {Xavier Gràcia and Josep M. Pons and Narciso Román-Roy},
    title = {{Higher-order Lagrangian systems: Geometric structures, dynamics, and constraints}},
    journal = {J. Math. Phys.},
    volume = {{\bf 32}},
    pages = {2744--2763},
    year = {1991},
    note = {\href{https://doi.org/10.1063/1.529066}{10.1063/1.529066}}
}

@article{GPR_92,
    author = {Xavier Gràcia and Josep M. Pons and Narciso Román-Roy},
    title = {{Higher-order conditions for singular Lagrangian dynamics}},
    journal = {J. Phys. A: Math. Gen.},
    volume = {{\bf 25}},
    number = {7},
    pages = {1989},
    year = {1992},
    note = {\href{https://doi.org/10.1088/0305-4470/25/7/037}{10.1088/0305-4470/25/7/037}}
}

@book{Her_30,
    author = {G. Herglotz},
    title = {{Berührungstransformationen}},
    year = {1930},
    publisher = {Lectures at the University of Göttingen}
}

@book{Her_85,
    author = {G. Herglotz},
    title = {{Vorlesungen \"uber die Mechanik der Kontinua}},
    series = {Teubner-Archiv zur Mathematik},
    volume = {3},
    publisher = {Teubner},
    address = {Leipzig},
    year = {1985},
    note = {\href{https://doi.org/10.1007/978-3-7091-9510-9}{10.1007/978-3-7091-9510-9}}
}

@article{Kam_82,
    author = {K. Kamimura},
    title = {{Singular Lagrangians and constrained Hamiltonian systems, generalized canonical formalism}},
    journal = {Nuovo Cim.},
    volume = {{\bf 68}},
    pages = {33--54},
    year = {1982},
    note = {\href{https://doi.org/10.1007/BF02888859}{10.1007/BF02888859}}
}

@book{KMS_93,
    author    = {Ivan Kol{\'a}\v{r} and Peter W. Michor and Jan Slov{\'a}k},
    title     = {{Natural Operations in Differential Geometry}},
    address   = {Berlin},
    isbn      = {978-3-642-08149-1},
    publisher = {Springer},
    series    = {Springer Monographs in Mathematics},
    year      = {1993},
    note      = {\href{https://doi.org/10.1007/978-3-662-02950-3}{10.1007/978-3-662-02950-3}}
}

@article{LGGMR_23,
    author = {Manuel de León and Jordi Gaset and Xavier Gràcia and Miguel C. Muñoz-Lecanda and Xavier Rivas},
    title = {{Time-dependent contact mechanics}},
    journal = {Monatsh. Math.},
    volume = {{\bf 201}},
    pages = {1149--1183},
    year = {2023},
    note = {\href{https://doi.org/10.1007/s00605-022-01767-1}{10.1007/s00605-022-01767-1}}
}

@article{LL_19,
    author = {Manuel de León and Manuel Laínz-Valcázar},
    title = {{Contact Hamiltonian systems}},
    journal = {J. Math. Phys.},
    volume = {{\bf 60}},
    number = {10},
    pages = {102902},
    year = {2019},
    note = {\href{https://doi.org/10.1063/1.5096475}{10.1063/1.5096475}}
}

@article{LL_19a,
    author = {Manuel de León and Manuel Lainz-Valcázar},
    title = {{Singular Lagrangians and precontact Hamiltonian systems}},
    journal = {Int. J. Geom. Methods Mod. Phys.},
    volume = {{\bf 16}},
    number = {10},
    pages = {1950158},
    year = {2019},
    note = {\href{https://doi.org/10.1142/S0219887819501585}{10.1142/S0219887819501585}}
}

@book{LM_87,
    author = {Paulette Libermann and Charles-Michel Marle},
    doi = {10.1007/978-94-009-3807-6},
    publisher = {Springer Dordrecht},
    title = {{Symplectic geometry and analytical mechanics}},
    year = {1987},
    volume = {35},
    series = {Mathematics and Its Applications},
    note = {\href{http://doi.org/10.1007/978-94-009-3807-6}{10.1007/978-94-009-3807-6}}
}

@article{LS_17,
    author = {Manuel de León and Cristina Sardón},
    title = {{Cosymplectic and contact structures for time-dependent and dissipative Hamiltonian systems}},
    journal = {J. Phys. A: Math. Theor.},
    volume = {{\bf 50}},
    number = {25},
    pages = {255205},
    year = {2017},
    note = {\href{https://doi.org/10.1088/1751-8121/aa711d}{10.1088/1751-8121/aa711d}}
}

@article{MR_18,
    author = {N. E. Martínez-Pérez and C. Ramírez},
    title = {{On the Lagrangian description of dissipative systems}},
    journal = {J. Math. Phys.},
    volume = {{\bf 59}},
    number = {3},
    year = {2018},
    pages = {032904},
    note = {\href{https://doi.org/10.1063/1.5004796}{10.1063/1.5004796}}
}

@article{Pon_88,
    author = {Josep M. Pons},
    title = {{New relations between Hamiltonian and Lagrangian constraints}},
    journal = {J. Phys. A: Math. Gen.},
    volume = {{\bf 21}},
    number = {12},
    year = {1988},
    pages = {2705},
    note = {\href{https://doi.org/10.1088/0305-4470/21/12/014}{10.1088/0305-4470/21/12/014}}
}

@book{Poo_81,
    author = {Walter A. Poor},
    title = {{Differential Geometric Structures}},
    year = {1981},
    publisher = {McGraw-Hill},
    address = {New York}
}

@article{RRS_05,
    author = {\'Angel M. Rey and Narciso Román-Roy and Modesto Salgado},
    title = {{Günther's formalism ($k$-symplectic formalism) in classical field theory: Skinner--Rusk approach and the evolution operator}},
    journal = {J. Math. Phys.},
    volume = {{\bf 46}},
    number = {5},
    pages = {052901},
    year = {2005},
    note = {\href{https://doi.org/10.1063/1.1876872}{10.1063/1.1876872}}
}

@article{Ski_83,
    author = {Ray Skinner},
    title = {First‐order equations of motion for classical mechanics},
    journal = {J. Math. Phys.},
    volume = {{\bf 24}},
    pages = {2581–-2588},
    year = {1983},
    note = {\href{https://doi.org/10.1063/1.525653}{10.1063/1.525653}}}

@article{Tul_1976b,
    author = {Wlodzimierz M. Tulczyjew},
    title = {{Les sous-variétés lagrangiennes et la dynamique lagrangienne}},
    journal = {C. R. Acad. Sc. Paris Sér. A},
    volume = {{\bf 283}},
    number = {8},
    pages = {675--678},
    year = {1976}
}
}

\end{document}